\documentclass[letterpaper,twocolumn,10pt,accepted=2021-05-08]{quantumarticle}
\pdfoutput=1
\usepackage[latin1]{inputenc}
\usepackage[english]{babel}
\usepackage[T1]{fontenc}
\usepackage{amsfonts,amssymb,amsmath,amsthm,bm,bbm}
\usepackage{hyperref}
\usepackage{microtype}

\theoremstyle{plain}
\newtheorem{prop}{Proposition}

\theoremstyle{remark}
\newtheorem{rmk}{Remark}
\newtheorem{ex}{Example}

\renewcommand{\H}{\mathcal{H}}
\renewcommand{\S}{\mathcal{S}}
\newcommand{\A}{\mathcal{A}}
\renewcommand{\P}{\mathcal{P}}
\newcommand{\U}{\mathcal{U}}
\newcommand{\p}{\mathfrak{p}}
\newcommand{\EE}{\mathbb{E}}
\newcommand{\tr}{\operatorname{tr}}
\renewcommand{\d}{\mathrm{d}}

\newcommand{\ds}{\mathrm{d}s}
\newcommand{\dist}{\operatorname{dist}}
\newcommand{\1}{\mathbbm{1}}
\newcommand{\bra}[1]{\langle #1|}
\newcommand{\ket}[1]{|#1\rangle}
\newcommand{\braket}[2]{\langle #1 | #2 \rangle}
\newcommand{\ketbra}[2]{|#1 \rangle\langle #2|}
\newcommand{\Log}{\operatorname{Log}}
\newcommand{\linspan}{\text{sp}}
\newcommand{\taumin}{\tau_{\text{pas}}}
\newcommand{\tauqsl}{\tau_{\text{qsl}}}
\newcommand{\tauamin}{\tau_{\text{apas}}}
\newcommand{\tauaqsl}{\tau_{\text{aqsl}}}
\newcommand{\tauNqsl}{\tau_{\text{cqsl}}}
\newcommand{\tauNpas}{\tau_{\text{cpas}}}

\begin{document}
\title{Time-optimal quantum transformations with bounded bandwidth}
\author{Dan Allan}
\author{Niklas H{\"o}rnedal}
\author{Ole Andersson}
\email{ole.andersson@fysik.su.se}
\affiliation{Fysikum, Stockholms universitet, 106 91 Stockholm, Sweden}

\begin{abstract}
In this paper, we derive sharp lower bounds, also known as quantum speed limits, for the time it takes to transform a quantum system into a state such that an observable assumes its lowest average value. We assume that the system is initially in an incoherent state relative to the observable and that the state evolves according to a von Neumann equation with a Hamiltonian whose bandwidth is uniformly bounded. The transformation time depends intricately on the observable's and the initial state's eigenvalue spectrum and the relative constellation of the associated eigenspaces. The problem of finding quantum speed limits consequently divides into different cases requiring different strategies. We derive quantum speed limits in a large number of cases, and we simultaneously develop a method to break down complex cases into manageable ones. The derivations involve both combinatorial and differential geometric techniques. We also study multipartite systems and show that allowing correlations between the parts can speed up the transformation time. In a final section, we use the quantum speed limits to obtain upper bounds on the power with which energy can be extracted from quantum batteries.
\end{abstract}

\maketitle

\section{Introduction}
A quantum speed limit (QSL) is a lower bound for the time it takes to transform a quantum state in a certain way under some given conditions. Many QSLs have been derived for both open and closed systems; see \cite{Fr2016, DeCa2017} and the references therein. Several of them are valid under very general conditions and can therefore be applied to virtually any system \cite{TaEsDadMF2013,dCaEgPlHu2013,ZhHaXiCaFa2014,DeLu2014,PiCiCeAdSo-Pi2016}. Extensive applicability is indeed a strength of a QSL but also means that the QSL can give a rather weak time-bound in specific~cases.

In this paper, we take a different approach to deriving QSLs for a family of quantum systems broad enough to include many systems of both practical and theoretical interest but narrow enough for the QSLs to be very sharp. Specifically, we consider a general finite-dimensional system prepared in a state that commutes with a definite but otherwise unspecified observable. For such a system, we examine how long it takes to unitarily transform the state into one where the observable's expectation value is minimal. We allow the Hamiltonian governing the transformation to be time-dependent, but we assume that its energy bandwidth is uniformly bounded. Such an assumption is in many cases physically justified \cite{CaHoKoOk2006,CaHoKoOk2007,CaHoKoOk2008,RuSt2014,RuSt2015,BrMe2015,BrGiMe2015,WaAlJaLlLuMo2015,GeWuWaXuShXiRoDu2016,WaKo2020}.

We borrow terminology from thermodynamics and call a state where the expectation value of the observable is minimal \emph{passive} \cite{PuWo1978,Le1978,SkSiBr2015}. A \emph{passivization process} is then any unitary process that leaves the system in a passive state. We define the \emph{passivization time} of the system to be the shortest time in which a passivization process that meets the bounded bandwidth condition can drive the system into a passive state. And by a QSL we henceforth mean a lower bound on the passivization~time. 

After a short preliminary section, we derive a general QSL for systems that conform to the description above. We also describe conditions that ensure that the QSL agrees with the passivization time. In connection with this, we study collective passivization processes of ensembles of identical systems. In such, we allow correlations to develop during the process. We show that allowing correlations reduces the passivization time of the individual system, and we extend the QSL to a lower bound on the collective passivization time.

For many systems, the general QSL is not tight. This is because the passivization time depends in a rather intricate way on the eigenvalue spectra and the eigenspaces of the observable and the initial state. We calculate the passivization time explicitly for completely non-degenerate systems. We also develop a method to derive QSLs for systems where the observable or the state has a degenerate spectrum. The method, which is particularly effective for low-dimensional systems, generates tight QSLs under additional conditions (precisely described in the paper).

Quantum batteries are systems of potentially great practical importance to which we can apply the results in the current paper. We demonstrate this by deriving sharp estimates for the power with which energy can be extracted from a quantum battery. Here, we follow \cite{AlBaNi2004,AlFa2013,BiViMoGo2015,CaPoBiCeGoViMo2017} and define a quantum battery as a closed quantum system whose energy content can be adjusted through cyclic unitary processes. 

Recent discoveries suggest that collective effects, such as entanglement, can improve the performance of quantum batteries \cite{AlFa2013,J-FSaRiBeLe2019}. For example, allowing correlations to develop between the batteries in an ensemble seems to increase the power with which energy can be extracted from each of the batteries \cite{BiViMoGo2015,CaPoBiCeGoViMo2017,J-FSaRiBeLe2019}. We use the QSLs for collective passivation processes to derive estimates for the power of collective energy extraction processes. Similar results but for different constraints can be found in \cite{BiViMoGo2015,CaPoBiCeGoViMo2017,J-FSaRiBeLe2019}.

The outline of the paper is as follows. In Section \ref{preliminaries} we set up the problem, introduce terminology, and prove some preliminary results. In Section \ref{a qsl} we derive a general QSL and discuss circumstances under which it is tight. In this section we also consider multipartite systems. Section \ref{the non-deg case} deals with the case when both the observable and the initial state have a non-degenerate spectrum. Section \ref{deg cases} begins with a general discussion on the characteristics of time-optimal Hamiltonians. Then we develop a method to deal with systems for which the observable or the initial state has a degenerate spectrum. In Section \ref{charging} we use the results from previous sections to derive upper bounds on the power of quantum batteries. The paper concludes with a summary and an outlook.

\section{Preliminaries}\label{preliminaries}
Let $A$ be an observable of an $n$-dimensional quantum system prepared in a state $\rho_i$. In this paper, we examine how long it takes before the system enters a state where the expectation value of $A$ is minimal. We assume that the system evolves according to a von Neumann equation $\dot\rho(t)=-i[H(t),\rho(t)]$ with a Hamiltonian satisfying the inequality
\begin{equation}\label{bbcondition}
	\tr\big( H(t)^2 \big) \leq \omega^2.
\end{equation}
The quantity on the left is \emph{the bandwidth} of the Hamiltonian, and the quantity $\omega$ on the right is some fixed positive number. We will refer to the inequality \eqref{bbcondition} as \emph{the bounded bandwidth condition}.

The eigenvalue spectrum of a quantum state that evolves according to a von Neumann equation is preserved. We write $\S(\rho_i)$ for the space of all states that have the same spectrum as $\rho_i$. Also, we write $\H$ for the Hilbert space of the system and denote the group of unitary operators on $\H$ by $\U(\H)$. The unitary group acts  on states according to $U\cdot\rho = U\rho U^\dagger$. This action preserves and is transitive on $\S(\rho_i)$. Since $\rho_i$ is assumed to evolve unitarily, $\S(\rho_i)$ can be considered the entire state-space of the system. For simplicity, ``state'' will from now on refer to a member of $\S(\rho_i)$ and thus be an abbreviation for ``state isospectral to $\rho_i$.''

The expectation value function of $A$ on $\S(\rho_i)$,
\begin{equation}
	\EE_A(\rho) = \tr(\rho A),
\end{equation} 
is real-valued and continuous. Since $\S(\rho_i)$ is compact and connected, the image of $\EE_A$ is a closed and bounded interval. We borrow terminology from thermodynamics and call the states at which $\EE_A$ assumes its minimum value \emph{passive} and the states at which $\EE_A$ assumes its maximum value \emph{maximally active}.\footnote{In thermodynamics an active state is any state which is not passive. The states that maximizes the expectation value of some observable (typically a Hamiltonian) are called maximally active.}  Then, a more appropriate formulation of the main question addressed in this paper is: What is the shortest time in which $\rho_i$ can be transformed into a passive state using a Hamiltonian that satisfies the bounded bandwidth condition?

\begin{rmk}
All results in this paper have a counterpart with an analogous proof for a maximally active final~state.
\end{rmk}

\subsection{Extremal and incoherent states}\label{extremals}
We call a state incoherent if it commutes with $A$.\footnote{This definition differs slightly from the usual one because $A$ is not required to have a non-degenerate spectrum.} An incoherent state is thus a state that preserves $A$'s eigenspaces. Since $\H$ is the direct sum of the eigenspaces of $A$, the incoherent states decompose into direct sums of operators acting on the eigenspaces. We arrange the different eigenvalues of $A$ in increasing order and write $\A_k$ for the eigenspace belonging to eigenvalue number $k$. A state $\rho$ is then incoherent if, and only if, $\rho = \oplus_k \rho_k$ where $\rho_k$ is an operator on $\A_k$. The operator $\rho_k$ is the \emph{$k$th component of $\rho$}.  

\begin{prop}\label{incoherent}
Passive and maximally active states are incoherent.
\end{prop}

\noindent The proposition is known since before \cite{AlBaNi2004}. But for the reader's convenience, we have included a proof in Appendix \ref{E}. In this paper, we assume the following:
\begin{center}
	\emph{The initial state $\rho_i$ is incoherent.}
\end{center}

The group of unitary operators commuting with $A$ acts transitively on the passive states. This means on the one hand that if $\rho$ is a passive state and $U$ is a unitary that commutes with $A$, then $U\rho U^\dagger$ is a passive state, and on the other hand that every passive state is of the form $U\rho U^\dagger$ for some unitary $U$ commuting with $A$. Since the unitaries commuting with $A$ decompose into direct sums of unitaries operating on the eigenspaces of $A$, all passive states have isospectral components. We conclude that a state is passive if, and only if, it is incoherent and has components isospectral with those of a passive state.

\subsection{Time-optimal Hamiltonians}
We define \emph{the passivization time} $\taumin$ as the shortest time in which $\rho_i$ can be transformed into a passive state using a Hamiltonian that satisfies the bounded bandwidth condition. Below we determine $\taumin$ in several important cases. We also give examples of Hamiltonians that realize such an optimal transformation. We say that a Hamiltonian is \emph{time-optimal} if it satisfies \eqref{bbcondition} and transforms $\rho_i$ into a passive state in the time $\taumin$. According to the next proposition, time-optimal Hamiltonians saturate the inequality \eqref{bbcondition} at all times.

\begin{prop}\label{saturates}
Time-optimal Hamiltonians saturate the bounded bandwidth condition.
\end{prop}

\noindent For a proof, see \cite{WaAlJaLlLuMo2015} or Appendix \ref{S}.

The time-evolution operator associated with a Hamiltonian can be considered as a curve in $\U(\H)$. We equip $\U(\H)$ with the bi-invariant Riemannian metric $g$ that agrees with the Hilbert-Schmidt inner product on the Lie algebra of $\U(\H)$. By Proposition \ref{saturates}, the time-evolution operator $U(t)$ associated with a time-optimal Hamiltonian $H(t)$ then has a constant speed $\omega$:
\begin{equation}
	g\big( \dot{U}(t), \dot{U}(t) \big) 
	= \tr\big( H(t)^2 \big) 
	= \omega^2.
\end{equation}

Let $\P(\rho_i)$ be the set of unitary operators that transform the initial state into a passive state:
\begin{equation}
	\P(\rho_i) = \{ U\in\U(\H) : U\rho_i U^\dagger \text{ is a passive state}\}.
\end{equation}
Proposition \ref{passivizing set} below says that $\P(\rho_i)$ is a submanifold of $\U(\H)$. The next proposition, proven in Appendix \ref{U}, transforms the problem of determining the passivization time into a geometric problem. We write $\1$ for the identity operator on $\H$.

\begin{prop}\label{shortest}
The time-evolution operator of a time-optimal Hamiltonian is a shortest curve from $\1$ to $\P(\rho_i)$.
\end{prop}

The shortest curves connecting $\1$ and $\P(\rho_i)$ are pre-geodesics (that is, they are geodesics if parameterized such that they have a constant speed \cite{Sa1996}). Every geodesic of $g$ emanating from $\1$ agrees with a one-parameter subgroup of $\U(\H)$ on its domain of definition. Conversely, every one-parameter subgroup of $\U(\H)$ is a geodesic of $g$; see \cite{Ar2003}. Thus, a curve of unitaries $U(t)$ such that $U(0)=\1$ is a geodesic if, and only if, $U(t)=e^{-itH}$ for some Hermitian operator $H$.~Proposition \ref{time-independent} is a direct consequence of Propositions \ref{saturates} and~\ref{shortest}.

\begin{prop}\label{time-independent}
Time-optimal Hamiltonians are time-independent.
\end{prop}

The geodesic distance between two unitary operators is the minimum of the lengths of all smooth curves connecting the two operators. One can express the geodesic distance between two unitaries $U$ and $V$ in terms of the principal logarithm and the Hilbert-Schmidt norm:
\begin{equation}\label{logdistance}
	\dist(U,V)
	= \|\Log (U^\dagger V)\| 
	= \sqrt{-\tr\big(\Log (U^\dagger V)^2\big)}.
\end{equation}
See Appendices \ref{L} and \ref{G} for details. It follows that
\begin{equation}\label{distance}
\begin{split}
	\dist\big(\1,\P(\rho_i)\big)
	&= \min\big\{ \dist(\1,U) : U\in\P(\rho_i)\big\} \\
	&= \min\big\{ \| \Log U\| : U\in\P(\rho_i)\big\}.
\end{split}
\end{equation}
The first identity is the definition of the distance between $\1$ and $\P(\rho_i)$. Equation \eqref{distance} and Propositions \ref{saturates} and \ref{shortest} together imply that
\begin{equation}\label{taumin}
	\taumin = \frac{1}{\omega}\min\big\{ \| \Log U\| : U\in\P(\rho_i)\big\}.
\end{equation}
Unfortunately, the minimum on the right is in general difficult to determine. But we will find explicit expressions for $\taumin$ in several important cases.

\subsection{Passivizing unitaries and isotropy groups}\label{puig}
As was mentioned above, the unitary group acts on states by left conjugation. We write $\U(\H)_{\rho_i}$ for the isotropy group of the initial state:
\begin{equation}
	\U(\H)_{\rho_i} = \{ U\in \U(\H) : U\rho_i U^\dagger =\rho_i\}.
\end{equation}
The unitary group also acts on observables by right conjugation, $U\cdot B=U^\dagger B U$, and we write $\U(\H)_{A}$ for the isotropy group of the observable $A$:
\begin{equation}
	\U(\H)_{A} = \{ U\in \U(\H) : U^\dagger A U =A\}.
\end{equation}

\begin{prop}\label{passivizing set}
The set $\P(\rho_i)$ is a submanifold of $\U(\H)$. Moreover, if $P$ is any unitary in $\P(\rho_i)$, then
\begin{equation}
	\P(\rho_i)=\{ UPV : U\in \U(\H)_{A}, V\in \U(\H)_{\rho_i}\}.
\end{equation}
\end{prop}

\noindent The proof is postponed to Appendix \ref{P}. Proposition \ref{passivizing set} implies that if $P$ is any passivizing unitary, then
\begin{equation}
    \dist(\1,\P(\rho_i))=\min_{U,V}\|\Log(UPV)\|.
\end{equation}
The minimum is over the $U$s in $\U(\H)_A$ and $V$s in~$\U(\H)_{\rho_i}$.

Since $\rho_i$ is incoherent, $A$ and $\rho_i$ are simultaneously diagonalizable. We fix a common orthonormal eigenbasis $\ket{1},\ket{2},\dots,\ket{n}$ of $A$ and $\rho_i$, and we write $a_1,a_2,\dots,a_n$ and $p_1,p_2,\dots, p_n$ for the associated eigenvalues. We furthermore assume that the ordering of the vectors in the basis, hereafter referred to as the \emph{computational basis}, is such that $a_1\leq a_2\leq \cdots \leq a_n$. This ordering is consistent with the prior ordering of the different eigenvalues of $A$ in Section \ref{extremals}.

\section{A quantum speed limit}\label{a qsl} 
In this section, we derive a QSL for the time it takes to passivize an incoherent state using a Hamiltonian that satisfies the bounded bandwidth condition. To this end, let $\delta_k$ be the number of eigenvalues that the $k$th component of $\rho_i$ does not have in common with the $k$th component of a passive state, and let $\delta$ be the sum of all the $\delta_k$s. We call $\delta$ the \emph{discrepancy} of the initial state, and we define the QSL $\tauqsl$ as
\begin{equation}\label{theqsl}
	\tauqsl = \frac{\pi\sqrt{\delta}}{2\omega}.
\end{equation}

\begin{prop}\label{lowerbound}
The QSL $\tauqsl$ bounds the passivization time from below.
\end{prop}

\begin{proof}
Let $H$ be a time-optimal Hamiltonian. For each computational basis vector $\ket{k}$ define $\ket{k(t)}=e^{-itH}\ket{k}$ and regard $\ket{k(t)}$ as a curve on the unit sphere in $\H$. The speed of $\ket{k(t)}$ equals $\bra{k}H^2\ket{k}^{1/2}$, and Proposition \ref{saturates} says that the squares of these speeds sum to $\omega^2$. Since the system evolves into a passive state in the time $\taumin$, at least $\delta$ of the basis vectors evolve, in this time, into eigenvectors of $A$ with a different eigenvalue and, thus, into orthogonal states. The spherical distance between any pair of orthogonal states is $\pi/2$. Hence,
\begin{equation}
	\omega^2\taumin^2=\sum_{k=1}^n \bra{k}H^2\ket{k} \taumin^2 
	\geq \frac{\pi^2\delta}{4}.
\end{equation}
This proves that $\taumin\geq\tauqsl$.
\end{proof}

A natural question is when the passivization time is equal to $\tauqsl$. Below we will see that such is the case if the initial state can be `permuted' to a passive state with a permutation whose cycles have a length of at most 2. We will also see examples of systems for which the passivization time is greater than  $\tauqsl$.

\subsection{Systems for which the passivization time equals the quantum speed limit}\label{agree}
For a not necessarily unique permutation $\sigma$ of the set $\{1,2,\dots,n\}$, the state $\rho_\sigma=\sum_k p_{\sigma(k)}\ketbra{k}{k}$ is passive. We call such a permutation \emph{passivizing}. Define the permutation operator associated with $\sigma$ as
\begin{equation}
	P_\sigma = \sum_{k=1}^n \ketbra{k}{\sigma(k)}.
\end{equation}
The operator $P_\sigma$ is unitary and $\rho_\sigma=P_\sigma\rho_iP_\sigma^\dagger$.

Any permutation of $\{1,2,\dots,n\}$ can be uniquely decomposed into disjoint cycles \cite{La1993}. Each cycle is itself a permutation of a subset of $\{1,2,\dots,n\}$. We denote the cycle which permutes the subset $\{k_1, k_2,\dots, k_l\}$ according to $k_1\to k_2\to \cdots\to k_l\to k_1$ by $(k_1\, k_2\,\dots\, k_l)$; see Figure \ref{cycle}.
\begin{figure}[t]
	\centering
	\includegraphics[width=0.20\textwidth]{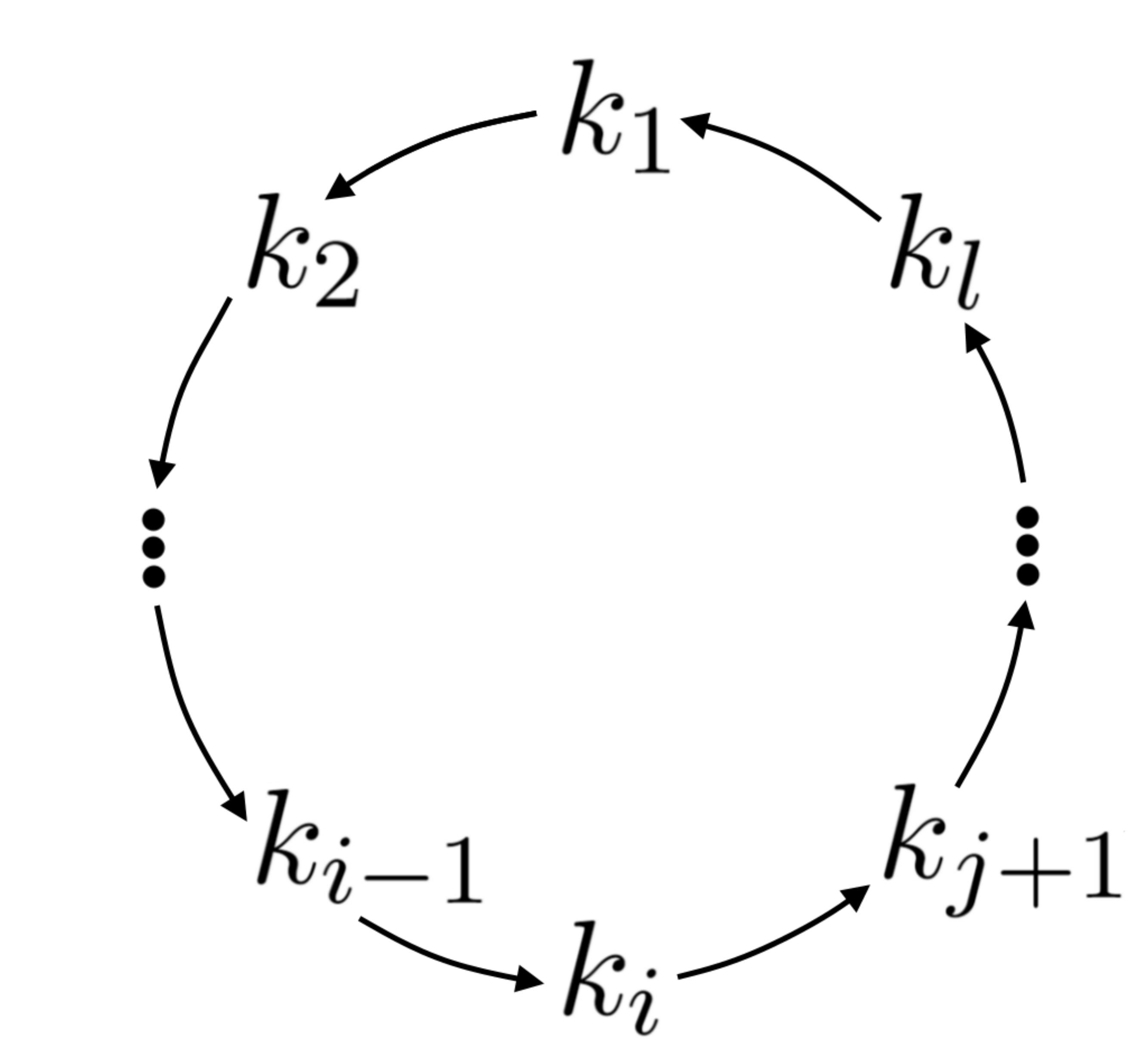}
	\caption{A graphical representation of $(k_1\,k_2\dots k_l)$.
	The number of elements in the cycle is the length of the cycle. A trivial cycle has length $1$ and a transposition has length $2$.} 
	\label{cycle}
\end{figure}
The number of elements $l$ is the length of the cycle. A cycle of length $1$ will be called trivial, and a cycle of length $2$ will be called a transposition. A permutation whose square leaves every element of $\{1,2,\dots,n\}$ invariant is called an involution. Being an involution is equivalent to having cycles of length at most $2$.

\begin{prop}\label{2cycles}
If $\rho_i$ can be passivized by an involution, then the passivization time equals $\tauqsl$.
\end{prop}

\begin{proof}
Let $\sigma$ be a passivizing involution. Reduce $\sigma$ by replacing each transposition $(k_1\,k_2)$ of $\sigma$ for which $a_{k_1}=a_{k_2}$ or $p_{k_1}=p_{k_2}$ holds by a pair of trivial cycles $(k_1)(k_2)$. The reduced $\sigma$ is also passivizing.

Let $c_1,c_2,\dots, c_m$ be the transpositions of the reduced $\sigma$. Write $c_j=(k^j_1\,k^j_2)$, with $k^j_1 < k^j_2$, and re-index the transpositions so that every $c_j$ for which $p_{k^j_1} < p_{k^j_2}$ holds has a lower index than all those $c_j$s for which $p_{k^j_1} > p_{k^j_2}$ holds. The Hamiltonian
\begin{equation}
	H=\frac{\omega}{\sqrt{2m}} \sum_{j=1}^m \big( \ketbra{k^j_2}{k^j_1} +\ketbra{k^j_1}{k^j_2} \big)
\end{equation} 
satisfies the bounded bandwidth condition and implements the passivizing unitary $-iP_\sigma$ in the time $\pi\sqrt{2m}/2\omega$. We will show that $\delta\geq 2m$. The opposite inequality follows from Proposition \ref{lowerbound}. 

For each $j$, let $P_{j}$ be the unitary operator which interchanges $\ket{k^j_1}$ and $\ket{k^j_2}$ and leaves all the other computational basis vectors invariant. Set $\rho_0=\rho_i$ and inductively define $\rho_j=P_{j}\rho_{j-1}P_{j}^\dagger$. The sequence of $\rho_j$s starts at $\rho_i$ and ends at the passive state $\rho_\sigma$. Moreover,
\begin{equation}\label{lessandless}
	\EE_A(\rho_j)=\EE_A(\rho_{j-1})+(a_{k^j_2}-a_{k^j_1})(p_{k^j_1}-p_{k^j_2}).
\end{equation}
The second term on the right is positive if $p_{k^j_1} > p_{k^j_2}$ and negative if $p_{k^j_1} < p_{k^j_2}$. The former situation is, however, excluded by the selected order of the transpositions. Otherwise, the final state is not passive. Hence, the sequence of expectation values $\EE_A(\rho_j)$ is monotonically decreasing. This, in turn, means that each  $p_{k^j_1}$ belongs to the spectrum of different components of $\rho_i$ and $\rho_\sigma$, and similarly for $p_{k^j_2}$. We conclude that $\delta\geq 2m$.
\end{proof}

In the proof of Proposition \ref{2cycles}, we chose to implement $-iP_\sigma$ rather than $P_\sigma$ because the latter does not belong to the passivizing unitaries closest to the identity. Explicitly, the distance from $\1$ to $P_\sigma$ equals $\pi\sqrt{m}$, which is $\sqrt{2}$ times greater than the distance from $\1$ to $\P(\rho_i)$.

\begin{ex}\label{exempt}
Suppose that $A$ has only two, possibly degenerate, eigenvalues. The incoherent states then have two components. Let $m$ be the number of eigenvalues that the first component of $\rho_i$ does not have in common with a passive state. Then the number of eigenvalues that the second component does not have in common with a passive state is also $m$, and $\delta=2m$. Match the `errant' eigenvalues in pairs such that each pair contains one eigenvalue from the first component and one eigenvalue from the second component of $\rho_i$. The indices of the eigenvalues in each pair define a transposition. Let $\sigma$ be the product of the so obtained transpositions times the trivial cycles whose elements are the indices of the unpaired eigenvalues. The permutation $\sigma$ is an involution, and the passivization time thus equals $\tauqsl$.
\end{ex}

The next example is interesting from a thermodynamic perspective. We shall return to this fact in Section \ref{charging}.

\begin{ex}\label{active} 
If $\rho_i$ is maximally active, the sequence of eigenvalues of $\rho_i$ is non-decreasing. The state with the reversed spectrum, $\rho_p=\sum_k p_{n-k+1}\ketbra{k}{k}$, is passive, and the discrepancy of $\rho_i$ equals $\delta=2m$ where $m$ is the greatest integer such that $p_m<p_{n-m+1}$ and $a_m<a_{n-m+1}$.\footnote{We assume that neither $A$ nor $\rho_i$ is proportional to $\1$ for then $\rho_i$ is already passive.} Defining $\sigma$ as
\begin{equation}
	\sigma = (1\; n)(2\; n-1) \cdots
	(\tfrac{\delta}{2}\; n-\tfrac{\delta}{2}+1)
	(\tfrac{\delta}{2}+1) \cdots (n-\tfrac{\delta}{2}),
\end{equation}
then $\sigma$ is passivizing and $\rho_p=P_\sigma\rho_i P_\sigma^\dagger$. The Hamiltonian
\begin{equation}
	H = \frac{\omega}{\sqrt{\delta}} \sum_{k=1}^{\delta/2} 
	\big( \ketbra{n-k+1}{k} + \ketbra{k}{n-k+1} \big)
\end{equation}
is time-optimal and implements $-iP_{\sigma}$ in the time $\tauqsl$.
\end{ex}

\subsection{Assisted passivization}
A \emph{passivization catalyst} is an auxiliary quantum system used to reduce the passivization time of a system. The catalyst is allowed to transform with the system. But as the system develops into a passive state, the catalyst must return to its original, uncorrelated state. As we will see, allowing correlations between the two in the meantime can reduce the passivization time.

To derive a bandwidth bound of the composite system that allows for a fair comparison between the transformation time of a catalyzed and that of an uncatalyzed passivizing transformation, consider a system in a state $\rho_i$ which is coupled to, but uncorrelated with, a catalyst in a state $\rho_c$. Assume that the system and the catalyst in the time $\tau$ evolve in parallel to $\rho_p\otimes\rho_c$, with $\rho_p$ being passive, according to a von Neumann equation with Hamiltonian $H_{sc}(t)=H_s(t)+ H_c(t)$. Furthermore, assume that the bandwidth of $H_s(t)$ is bounded from above by $\omega^2$. Then $\tau$ is greater than or equal to the system's passivization time $\taumin$. And only if $H_s(t)$ is time-optimal and $H_c(t)$ is suitably adjusted, $\tau$ can be equal to $\taumin$. In that case, the bandwidth of $H_{sc}(t)$ is $n_c\omega^2+n\tr(H_c(t)^2)$, where $n_c$ and $n$ is the dimension of the catalyst and the system, respectively. Here we have used that time-optimal Hamiltonians are traceless; see Section \ref{traceless}. We thus formulate the bounded bandwidth condition for assisted transformations as
\begin{equation}\label{assbbc}
	\tr\big(H_{sc}(t)^2\big)\leq n_c\omega^2.
\end{equation}

We define \emph{the assisted passivization time} $\tauamin$ as the shortest time in which a system can be transformed into a passive state in a catalyzed process governed by a Hamiltonian satisfying \eqref{assbbc}. Moreover, for a system in a state with discrepancy $\delta$, we define \emph{the assisted QSL}~as
\begin{equation}
	\tauaqsl=\frac{\pi}{2\omega}\sqrt{\frac{\delta}{n_c}}.
\end{equation}
The next two propositions say that the assisted passivization time is at least $\tauaqsl$ but not greater than $\taumin/\sqrt{n_c}$.

\begin{prop}\label{assistedupperbound}
The assisted passivization time is not greater than $\taumin/\sqrt{n_c}$.
\end{prop}

\begin{proof}
Let $H_s$ be a time-optimal Hamiltonian for the system, transforming $\rho_i$ into $\rho_p$ in the time $\taumin$, and let $\ket{\psi}$ be the pure state of an $n_c$-dimensional catalyst. Define a Hamiltonian for the combined system and catalyst as $H_{sc}=\sqrt{n_c}H_s\otimes\ketbra{\psi}{\psi}$. The Hamiltonian $H_{sc}$ has bandwidth $n_c\omega^2$ and transforms $\rho_i\otimes\ketbra{\psi}{\psi}$ into $\rho_p\otimes\ketbra{\psi}{\psi}$ in the time $\taumin/\sqrt{n_c}$.
\end{proof}

\begin{prop}\label{assistedlowerbound}
The assisted QSL $\tauaqsl$ lower bounds the assisted passivization time.
\end{prop}

\begin{proof}
Let $\rho_i$ and $\rho_p$ be the initial and a passive state of the system, respectively, and let $\rho_c$ be the state of an $n_c$-dimensional catalyst. Since $\rho_i$ and $\rho_p$ are incoherent relative to $A$, the product states $\rho_i\otimes \rho_c$ and $\rho_p\otimes \rho_c$ are incoherent relative to $A\otimes \1$. Moreover, the $k$th components of $\rho_i\otimes \rho_c$ and $\rho_p\otimes \rho_c$ are $\rho_{i;k}\otimes \rho_c$ and $\rho_{p;k}\otimes \rho_c$, respectively, where $\rho_{i;k}$ and $\rho_{p;k}$ are the $k$th components of $\rho_i$ and $\rho_p$. 
	
Let $\delta_k$ be the discrepancy between $\rho_{i;k}$ and $\rho_{p;k}$ and let $\delta^c_k$ be the discrepancy between $\rho_{i;k}\otimes \rho_c$ and $\rho_{p;k}\otimes \rho_c$. Then $\delta^c_k\geq \delta_k$. To see this, let $p_{i_1}, p_{i_2},\dots, p_{i_{\delta_k}}$ be the eigenvalues of $\rho_{i;k}$ that are not present in the spectrum of $\rho_{p;k}$. Each of these eigenvalues is either strictly greater than or strictly smaller than all of $\rho_{p;k}$'s eigenvalues. Otherwise, $\rho_p$ would not be passive. Organize the differing eigenvalues of $\rho_{i;k}$ so that the first $l$ ones are greater than and the last $\delta_k-l$ ones are smaller than all the eigenvalues of $\rho_{p;k}$. Furthermore, let $q_1$ be the greatest and $q_2$ be the smallest non-zero eigenvalue of $\rho_c$. (If $\rho_c$ is a pure state, then set $q_1=q_2=1$.) Neither of the eigenvalues $p_{i_1}q_1, p_{i_2}q_1, \dots, p_{i_{l}}q_1,p_{i_{l+1}}q_2, p_{i_{l+2}}q_2, \dots, p_{i_{\delta_k}}q_2$ of $\rho_{i;k}\otimes \rho_c$ is then present in the spectrum of $\rho_{p;k}\otimes \rho_c$. These are $\delta_k$ in number and, hence, $\delta^c_k\geq\delta_k$. 

Let $H_{sc}(t)$ be a Hamiltonian that satisfies \eqref{assbbc} and which, among such Hamiltonians, transforms $\rho_i\otimes \rho_c$ into $\rho_p\otimes \rho_c$ in the shortest time. (Without loss of generality, we can assume that this time equals $\tauamin$ since $\rho_p$ is unspecified.) Using arguments identical to the ones used in the uncatalyzed case, one can show that $\tr(H_{sc}(t)^2)=n_c\omega^2$ for all $t$ and that $H_{sc}(t)$ is time-independent.

Write $\delta^c$ for the sum of all the $\delta^c_k$s. Let $\ket{k\,l}$ be the product of the $k$th vector in the computational basis for the system and the $l$th vector in an eigenbasis of $\rho_c$. The trajectory formed when $\ket{k\,l}$ is affected by $H_{sc}$ has the constant speed $\bra{k\,l}H_{sc}^2\ket{k\,l}^{1/2}$. Furthermore, for at least $\delta^c$ such vectors, the trajectory has a length greater than or equal to $\pi/2$. Consequently,
\begin{equation}
	\omega^2n_c\tauamin^2
	\geq \sum_{k=1}^{n_s}\sum_{l=1}^{n_c}\bra{k\,l}H_{sc}^2\ket{k\,l}\tauamin^2
	\geq \frac{\pi^2\delta^c}{4}.
\end{equation}
The proposition follows from the inequality $\delta^c\geq \delta$.
\end{proof}

By Propositions \ref{assistedupperbound} and \ref{assistedlowerbound}, $\tauamin=\tauaqsl$ if $\taumin=\tauqsl$. Such is the case, for example, if the initial state can be passivized by an involution.

\begin{ex}
A system in a maximally active state can be passivized in the time $\tauqsl$ and be assisted passivized in the time $\tauqsl/\sqrt{n_c}$ using an $n_c$-dimensional catalyst.
\end{ex}	

\subsection{Collective passivization}\label{collpas}
Assume that $A$ has a non-degenerate spectrum. Then there is but one passive state $\rho_p$. In this section, we consider $N$ copies of the system prepared in the product state $\rho_i^{\otimes N}=\rho_i\otimes\rho_i\otimes\dots\otimes\rho_i$, and we ask if it is possible to transform $\rho_i^{\otimes N}$ into $\rho_p^{\otimes N}=\rho_p\otimes\rho_p\otimes\dots\otimes\rho_p$ in a time shorter than the single copy passivization time $\taumin$ using a Hamiltonian that satisfies the bandwidth condition
\begin{equation}\label{collectivebbcondition}
	\tr\big(H(t)^2\big)\leq \omega^2 N n^{N-1}.
\end{equation}
The right-hand side equals the bandwidth of an $N$-fold sum of local time-optimal Hamiltonians satisfying \eqref{bbcondition}. We will see that, as in the case of assisted passivization, allowing correlations between the systems can reduce the passivization time to a value smaller than $\taumin$.\footnote{This is true also when all the states that are unitarily equivalent to $\rho_i^{\otimes N}$ are separable \cite{GuBa2002,GuBa2003}.}

\begin{rmk}
Here we consider transformations of $\rho_i^{\otimes N}$ into a specific final state, namely $\rho_p^{\otimes N}$, rather than ``some passive state.'' The state $\rho_p^{\otimes N}$ need not be passive for either $N\cdot A=A+A+\cdots+A$ or $A^{\otimes N}=A\otimes A\otimes \cdots\otimes A$; see \cite{AlFa2013} and Example \ref{neithernor}.
\end{rmk}

The computational basis determines a canonical basis in the $N$-fold tensor product of $\H$. We write $\ket{k_1 k_2 \dots k_N}$ for $\ket{k_1}\otimes \ket{k_2}\otimes\dots\otimes\ket{k_N}$. Then
\begin{align}
	&\rho_i^{\otimes N}=\sum p_{k_1}\!\cdots p_{k_N}\ketbra{k_1 \dots k_N}{k_1 \dots k_N},\label{jaha}\\
	&\rho_p^{\otimes N}=\sum p_{\sigma(k_1)}\!\cdots p_{\sigma(k_N)}\ketbra{k_1 \dots k_N}{k_1 \dots k_N},\label{jahaja}
\end{align}
where $\sigma$ is any permutation that passivizes $\rho_i$.

The sums in Equations \eqref{jaha} and \eqref{jahaja} are over all the sequences $k_1,k_2,\dots,k_N$ one can form from the integers $1,2,\dots,n$. Let $\delta_N$ be the number of such sequences for which $p_{k_1}p_{k_2}\!\cdots p_{k_N}$ and $p_{\sigma(k_1)}p_{\sigma(k_2)}\!\cdots p_{\sigma(k_N)}$ are different. Also, define \emph{the collective passivization time} $\tauNpas$ as the minimum time it takes to transform $\rho_i^{\otimes N}$ into $\rho_p^{\otimes N}$ using a Hamiltonian that satisfies \eqref{collectivebbcondition}. Then
\begin{equation}\label{Nqsl}
	\tauNpas\geq\frac{\pi}{2\omega}\sqrt{\frac{\delta_N}{Nn^{N-1}}}.
\end{equation}
The proof is similar to that of Proposition \ref{lowerbound}: A Hamiltonian that meets the condition \eqref{collectivebbcondition} and transforms $\rho_i^{\otimes N}$ into $\rho_p^{\otimes N}$ in the time $\tauNpas$ is time-independent. Let $H$ be such a Hamiltonian. Then $H$ transforms each product basis vector $\ket{k_1k_2\dots k_N}$ into an eigenvector of $\rho_p^{\otimes N}$. Furthermore, if $p_{k_1}p_{k_2}\cdots p_{k_N}$ and $p_{\sigma(k_1)}p_{\sigma(k_2)}\cdots p_{\sigma(k_N)}$ are different, the length of the trajectory of this vector is at least $\pi/2$. Since the trajectory has the constant speed $\bra{k_1k_2\dots k_N}H^2\ket{k_1k_2\dots k_N}^{1/2}$,
\begin{equation}
	\omega^2 N n^{N-1}\tauNpas^2\geq\tr\big(H^2\big)\tauNpas^2\geq \frac{\pi^2\delta_N}{4}.
\end{equation}

The expression on the right-hand side of \eqref{Nqsl} is \emph{the collective QSL}, which we denote by $\tauNqsl$. The collective passivization time equals $\tauNqsl$ if $\sigma$ is an involution. Because if such is the case, the Hamiltonian
\begin{equation}
\begin{split}
	H = \frac{\pi}{2\tauNqsl} 
		\sum \big(&\ketbra{k_1 \dots k_N}{\sigma(k_1) \dots \sigma(k_N)}\\
		&+\ketbra{k_1 \dots k_N}{\sigma(k_1) \dots \sigma(k_N)}\big),
\end{split}
\end{equation}
where the sum is over all sequences $k_1,k_2,\dots,k_N$ for which $p_{k_1}p_{k_2}\!\cdots p_{k_N}$ and $p_{\sigma(k_1)}p_{\sigma(k_2)}\!\cdots p_{\sigma(k_N)}$ are different, satisfies \eqref{collectivebbcondition} and transforms $\rho_i^{\otimes N}$ into $\rho_p^{\otimes N}$ in the time $\tauNqsl$.

Next, we will calculate the fraction between the single system passivization time and the collective $N$-fold passivization time for systems prepared in maximally active qubits or qutrits. The fraction can be considered as a measure of the advantage of a collective passivization \cite{CaPoBiCeGoViMo2017}. For mixed qubits, the fraction depends explicitly on the parity of $N$, and to simplify the notation we will make use of the parity function 
\begin{equation}\label{parityfunk}
	\wp(k)=\frac{1}{2}\big(1+(-1)^k\big)=\begin{cases}
		0 & \text{if $k$ is odd},\\
		1 & \text{if $k$ is even}.
	\end{cases}
\end{equation}

\begin{ex}\label{collectivequbit}
Suppose that $\rho_i$ is a maximally active qubit state. If $\rho_i$ is pure, then $\delta_N=2$ and
\begin{equation}
	\frac{\taumin}{\tauNpas}=\sqrt{N2^{N-1}}.
\end{equation}
If $\rho_i$ mixed, then $\delta_N=2^N-\wp(N)\binom{N}{N/2}$ and
\begin{equation}\label{Advantage_qubits}
	\frac{\taumin}{\tauNpas}=\sqrt{\frac{N2^N}{2^N-\wp(N)\binom{N}{N/2}}}.
\end{equation}
The formula for $\delta_N$ follows immediately from the observation that $p_1^kp_2^{N-k}=p_2^kp_1^{N-k}$ if, and only if, $2k=N$.

It is apparent that allowing correlations between the qubits during the evolution reduces the passivization time. In Figure \ref{Fig:Qubit},
\begin{figure}[t]
	\centering
	\includegraphics[width=0.45\textwidth]{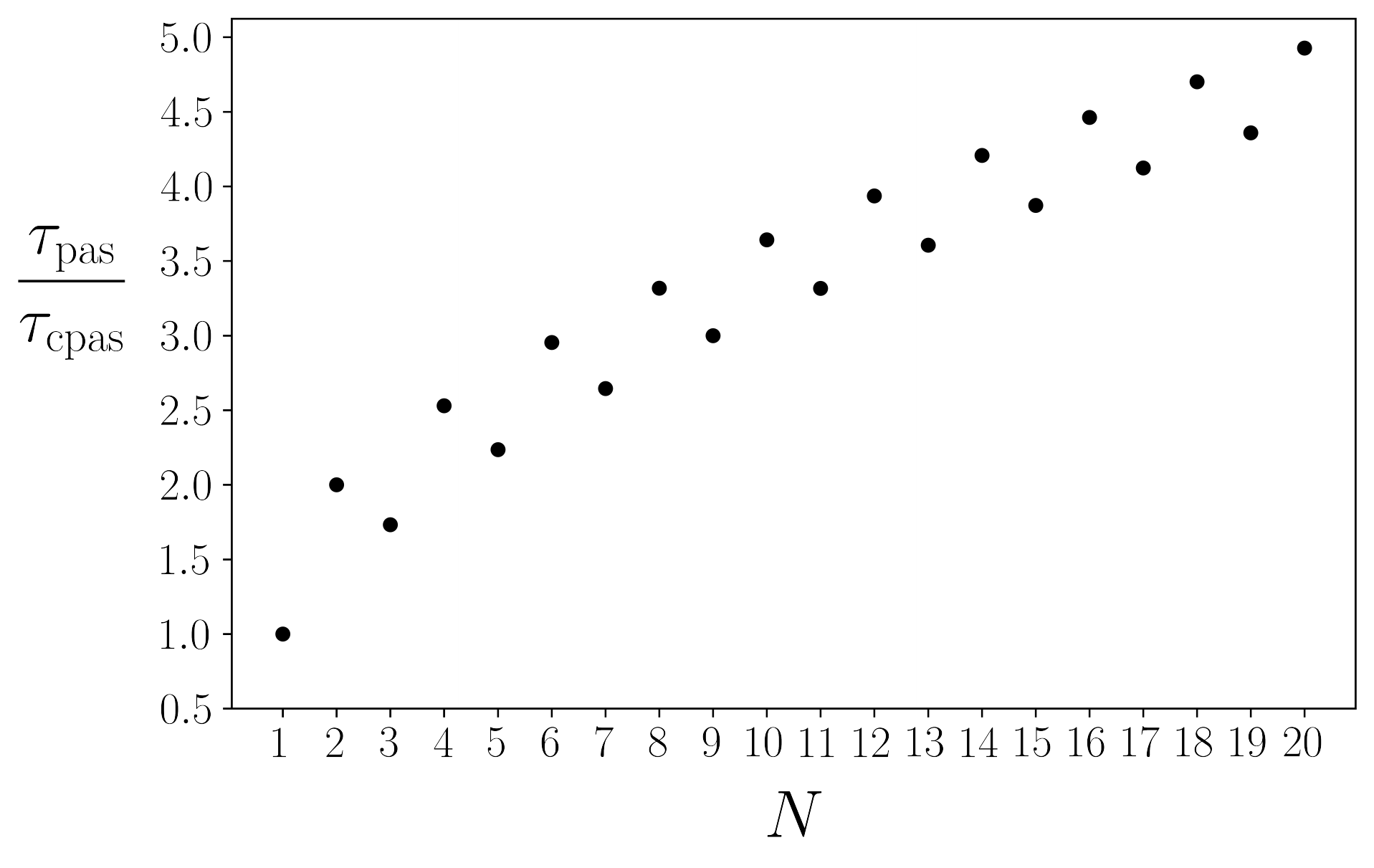}
	\caption{A plot of $\taumin/\tauNpas$ for an ensemble of $N$ mixed qubits. The trend of the plot tells us that the collective passivization time decreases with $N$, but the plot also has a noticeable fluctuating appearance indicating a non-monotonic dependence on $N$; if $N$ is even, the collective passivization time is smaller for $N$ qubits than for $N+1$ qubits but larger than for $N+2$ qubits.} 
	\label{Fig:Qubit}
\end{figure}
we have plotted $\taumin/\tauNpas$ against $N$ for a mixed $\rho_i$. The trend says that the more qubits are involved, the smaller the collective passivization time. However, as the plot also indicates, the decrease in collective passivization time is not monotonic in $N$. Adding a qubit to an ensemble with an even number of qubits increases the passivization time while adding a qubit to an ensemble with an odd number of qubits reduces the passivization time. Another interesting observation is that the asymptotic behavior of $\taumin/\tauNpas$ is very different for a pure and a mixed $\rho_i$. 
\end{ex}

\begin{ex}\label{collectivequtrit}
Suppose that $\rho_i$ is a maximally active qutrit state. If $p_1=0$, then $\delta_N=2(2^N-1)$ and 
\begin{equation}
    \frac{\taumin}{\tauNpas} = \sqrt{\frac{N3^{N-1}}{2^N-1}}.
\end{equation}
If $\rho_i$ has full rank, then $\delta_N=3^N-\sum_{k=0}^{\lfloor N/2\rfloor}\binom{N}{k,k}$ and 
\begin{equation}\label{Advantage_qutrits}
	\frac{\taumin}{\tauNpas}
	= \sqrt{\frac{2N3^{N-1}}{3^N-\sum_{k=0}^{\lfloor N/2\rfloor} \binom{N}{k,k}}}.
\end{equation}
The upper limit $\lfloor N/2\rfloor$ is the  greatest integer less than or equal to $N/2$, and $\binom{N}{k,k}$ is the trinomial coefficient $N!/k!k!(N-2k)!$. Figure \ref{Fig:Qutrit}
\begin{figure}[t]
	\centering
	\vspace{-2.5pt}
	\includegraphics[width=0.45\textwidth]{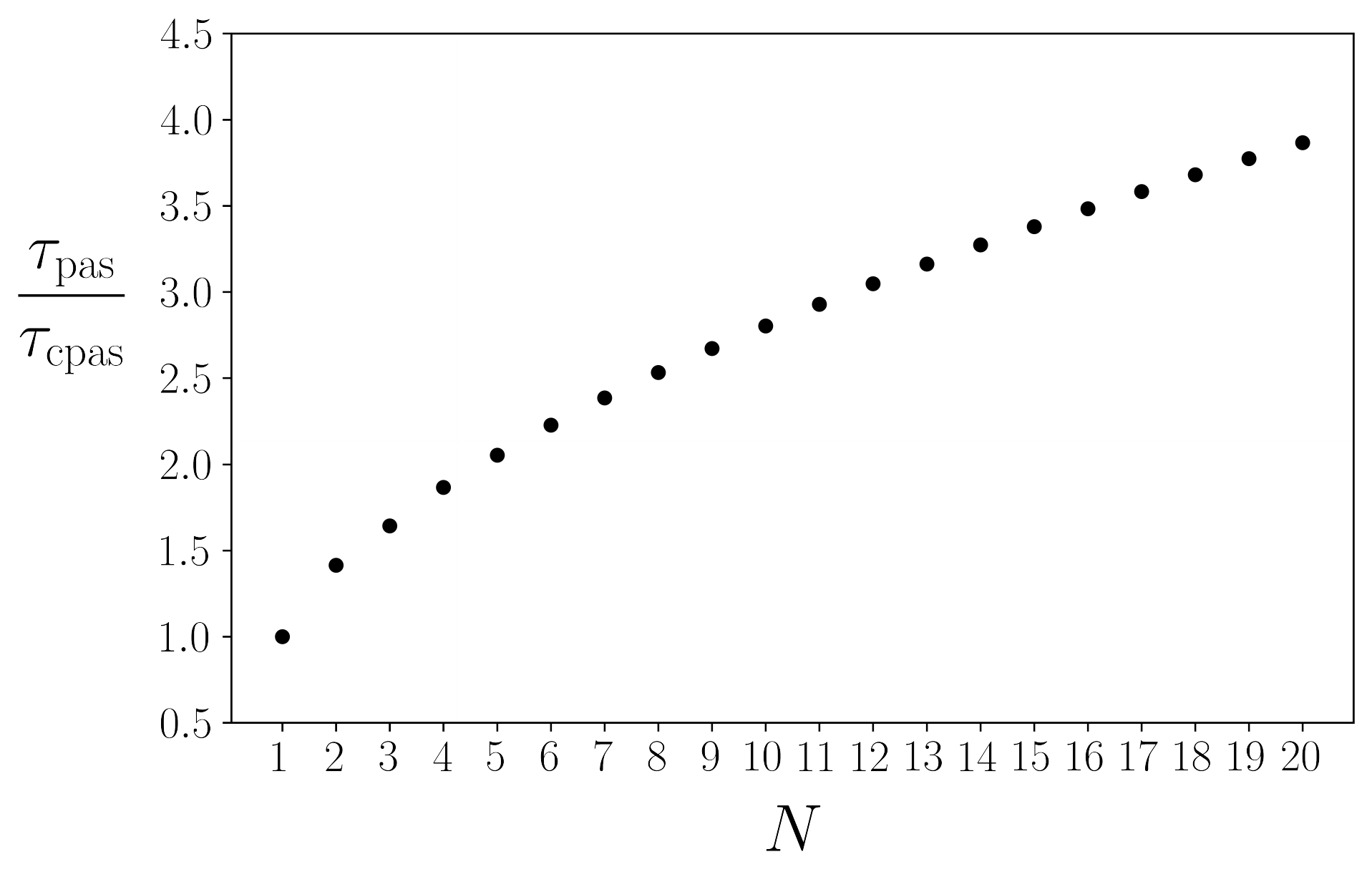}
	\caption{A plot of $\taumin/\tauNpas$ for an ensemble of $N$ full rank qutrits. Unlike the case for mixed qubits, the collective passivization time decreases monotonically with $N$.}
	\label{Fig:Qutrit}
\end{figure}
shows a plot of $\taumin/ \tauNpas$ against $N$ for a full rank $\rho_i$. In this case, the collective passivization time \emph{is} monotonic in $N$.
\end{ex}
 
For non-maximally active qutrit states, $\delta_N$ will depend on the spectrum. This is also the case for higher dimensional systems, even for maximally active initial states. We finish this section with two examples concerning the passivity of $\rho_p^{\otimes N}$. To distinguish the multipartite system's passivization time from that of the single-part system we call the former \emph{the global passivization time}.

\begin{ex}
If $\rho_i$ is a maximally active qubit, then $\rho_i^{\otimes N}$ is maximally active and $\rho_p^{\otimes N}$ is passive for both $N\cdot A$ and $A^{\otimes N}$. Furthermore, the collective passivization time is equal to the global passivization time.
\end{ex}

\begin{ex}\label{neithernor}
Suppose that $\rho_i$ is a maximally active qutrit state whose spectrum is such that
\begin{equation}
	p_1^2 < p_1p_2 < p_2^2 < p_1p_3 < p_2p_3 < p_3^2.
\end{equation}
Furthermore, suppose that the spectrum of $A$ is such~that 
\begin{align}
	&2a_1<a_1+a_2<2a_2<a_1+a_3<a_2+a_3<2a_3,\\
	&a_1^2<a_1a_2<a_2^2<a_1a_3<a_2a_3<a_3^2.
\end{align} 
Then $\rho_i^{\otimes 2}$ is a maximally active state for both $2\cdot A$ and $A^{\otimes 2}$. However, $\rho_p^{\otimes 2}$ is not a passive state for either $2\cdot A$ or $A^{\otimes 2}$. The discrepancy of $\rho_i^{\otimes 2}$ is $8$ and, hence, the global passivization time is $\pi/\omega\sqrt{3}$. Also, $\delta_2=6$ and $\tauNqsl=\pi/2\omega$. The collective passivization time is thus smaller than the global passivization time.
\end{ex}

\section{The non-degenerated case}\label{the non-deg case}
Suppose that $A$ and $\rho_i$ have non-degenerate eigenvalue spectra. Let $\rho_p$ be the unique passive state and let $\sigma$ be the unique permutation such that $\rho_p=\sum_k p_{\sigma(k)}\ketbra{k}{k}$. In this case, the lengths of the cycles of $\sigma$ determine the passivization time. 

\begin{prop}\label{nondegenerate}
Suppose that $\sigma$ decomposes into $m$ cycles of lengths $l_1,l_2,\dots,l_m$. Then
\begin{equation}\label{bound}
	\taumin = \frac{\pi}{\sqrt{3}\omega}\sqrt{n-\sum_{j=1}^m \frac{1}{l_j}}.
\end{equation}
\end{prop}

The proposition covers, for example, the case when the observable $A$ is non-degenerate, the initial state is prepared by measuring $A$, and the outcomes of the measurement are obtained with different frequencies.

\begin{ex}
Proposition \ref{nondegenerate} suggests that there are systems for which the passivization time is greater than $\tauqsl$. Consider, for example, a qutrit system for which $a_1<a_2<a_3$ and $p_2<p_1<p_3$ hold. The passive state is $\rho_p=p_3\ketbra{1}{1}+p_1\ketbra{2}{2}+p_2\ketbra{3}{3}$, and the permutation that transforms $\rho_i$ to $\rho_p$ is the $3$-cycle $(1,3,2)$. According to Proposition \ref{nondegenerate}, $\taumin=\pi\sqrt{8}/3\omega$, but $\tauqsl=\pi\sqrt{3}/2\omega$. Thus $\taumin>\tauqsl$. A time-optimal Hamiltonian that transforms $\rho_i$ to $\rho_p$ is 
\begin{equation}
	H=\frac{i\omega}{\sqrt{6}}
	\big( (\ket{2}-\ket{3})\bra{1} + (\ket{3}-\ket{1})\bra{2} + (\ket{1}-\ket{2})\bra{3} \big).
\end{equation}
\end{ex}

\begin{proof}[Proof of Proposition \ref{nondegenerate}]
Let $c_1,c_2,\dots,c_m$ be the cycles of $\sigma$ and let $\H_j$ be the linear span of those vectors in the computational basis whose labels are permuted by $c_j$. The $\H_j$s are mutually orthogonal and span $\H$. Also, the $\H_j$s are invariant for $P_\sigma$ and the isotropy groups of $A$ and $\rho_i$. Let $P_{c_j}$ be the restriction of $P_\sigma$ to $\H_j$.  

Every unitary $U$ that commutes with $A$ and every unitary $V$ that commutes with $\rho_i$ decomposes as $U=\oplus_j U_j$ and $V=\oplus_j V_j$, respectively, with $U_j$ and $V_j$ being operators of $\H_j$. Moreover,
\begin{equation}\label{enkelt}
	\|\Log(UP_\sigma V)\|^2 = \sum_{j=1}^m \|\Log(U_jP_{c_j} V_j)\|^2.
\end{equation}
According to \eqref{taumin} and Proposition \ref{passivizing set}, we can determine the passivization time by minimizing the terms on the right-hand side of \eqref{enkelt}.

The non-degeneracy of $A$ and $\rho_i$ implies that $U_j$ and $V_j$ are diagonal relative to the computational basis vectors that span $\H_j$. Write $c_j=(k_1, k_2, \dots, k_{l_j})$ and let $e^{i\alpha_r}$ be the eigenvalue of $U_j$, and $e^{i\beta_r}$ be the eigenvalue of $V_j$, associated with $\ket{k_r}$. The minimal polynomial of $U_jP_{c_j}V_j$ is $x^{l_j}-e^{i\theta_j}$ where $\theta_j=\sum_r(\alpha_r+\beta_r)\bmod{2\pi}$. Hence the eigenvalues of $U_jP_{c_j}V_j$ are $\lambda_r=e^{i(\theta_j+2\pi r)/l_j}$, where $r$ runs from $0$ to $l_j-1$. It follows that 
\begin{equation}
\begin{split}
	&\hspace{-3pt}\min_{U_j,V_j}\|\Log(U_jP_{c_j} V_j)\|^2 
	= \min_{\theta_j}\frac{1}{l_j^2} \sum_{r=0}^{l_j-1} (\theta_j+ 2\pi r)^2 \\
	&\hspace{-3pt}= \min_{\theta_j}\frac{1}{l_j} \Big(\hspace{-1pt}\theta_j^2 + 2\pi(l_j-1)\theta_j + \frac{2\pi^2}{3}(2l_j^2-3l_j+1)\hspace{-1pt}\Big)\\
	&\hspace{-3pt}=\frac{\pi^2(l_j^2-1)}{3l_j}.
\end{split}
\end{equation}
The minimum is attained for $\theta_j=\pi(1-l_j)$, which also meets the requirement that all the phases $(\theta_j+2\pi r)/l_j$ belong to the principal branch. We conclude that
\begin{equation}
	\taumin
	= \frac{\pi}{\omega}\sqrt{\sum_{j=1}^m \frac{l_j^2-1}{3l_j}}
	= \frac{\pi}{\sqrt{3}\omega}\sqrt{n-\sum_{j=1}^m \frac{1}{l_j}}.
\end{equation}
This proves Proposition \ref{nondegenerate}.
\end{proof}

The proof shows that $U=\oplus_je^{i\pi(1-l_j)/l_j}P_{c_j}$ is among the passivizing unitary operators closest to the identity. A time-optimal Hamiltonian that implements $U$ is 
\begin{equation}
	\hspace{-2pt}H=\omega\sqrt{\frac{3}{n-\sum_{k=1}^m\frac{1}{l_k}}} 
	\bigoplus_{j=1}^m \bigg( \frac{i}{\pi}\Log P_{c_j} + \frac{\wp(l_j)}{l_j}\1_j \bigg).
\end{equation}
The operator $\1_j$ is the identity operator on $\H_j$ and $\wp$ is the parity function defined in \eqref{parityfunk}.

\section{On degenerated cases}\label{deg cases}
In this section, we describe some general properties of time-optimal transformations. We also discuss circumstances under which these properties are sufficient to determine a degenerate system's passivization time.

\subsection{Incompatibility and parallelism of time-optimal Hamiltonians}\label{traceless}
In all the cases considered so far, the specified time-optimal Hamiltonians are traceless, and the passivizing unitaries lying closest to the identity are special unitary, that is, have determinant equal to $1$. As we will see, these observations are consequences of time-optimal Hamiltonians generating shortest curves between $\1$ and the manifold of passivizing unitaries $\P(\rho_i)$.

We say that a Hamiltonian is \emph{completely incompatible} with $A$ if $\Pi H\Pi=0$ for all the eigenspace projectors $\Pi$ of $A$. Moreover, using terminology from the theory of fiber bundles \cite{KoNo1996,AnHe2014,An2018}, we say that $H$ is \emph{parallel transporting} if $\Pi H\Pi=0$ for every eigenspace projector $\Pi$ of $\rho_i$. Both of these properties separately imply that $H$ is traceless and, hence, the time-evolution operator associated with $H$ is special unitary. 

\begin{prop}\label{parallel}
Time-optimal Hamiltonians are parallel transporting and completely incompatible with $A$.
\end{prop}

\begin{proof}
According to Proposition \ref{shortest}, time-optimal Hamiltonians generate shortest curves between $\1$ and $\P(\rho_i)$. Such a shortest curve has to meet $\P(\rho_i)$ perpendicularly \cite{Sa1996}. Let $H$ be a time-optimal Hamiltonian and let $U$ be the passivizing unitary generated in the time $\taumin$. The velocity vector at $U$ of the time-evolution operator of $H$ is $-iHU$. Let $B_1$ and $B_2$ be any Hermitian operators commuting with $A$ and $\rho_i$, respectively. By Proposition \ref{passivizing set}, $-iB_1U$ and $-iUB_2$ are tangent vectors of $\P(\rho_i)$ at $U$. These vectors are perpendicular to $-iHU$ and, hence, 
\begin{align}
	\tr(HB_1) &= g(-iHU,-iB_1U) = 0,\\
	\tr(HB_2) &= g(-iHU,-iUB_2) = 0.
\end{align}
Since $B_1$ and $B_2$ are arbitrary, the former identity implies that $H$ is completely incompatible with $A$, and the latter implies that $H$ is parallel transporting.
\end{proof}

\subsection{Upper bounds on the passivization time from passivizing permutations}
When $A$ or $\rho_i$ has a degenerate spectrum, the expression on the right-hand side of \eqref{bound} need not be equal to the passivization time. However, the expression always is an upper bound for the passivization time. To see this, let $\sigma=c_1c_2\cdots c_m$ be a passivizing permutation. For any sequence of real numbers $\theta_1,\theta_2,\dots,\theta_m$, the operator $\oplus_j e^{i\theta_j} P_{c_j}$ is a passivizing unitary. Therefore, by \eqref{taumin},
\begin{equation}\label{eqa}
	\taumin\leq \frac{1}{\omega} \sqrt{\sum_{j=1}^m \|\Log(e^{i\theta_j} P_{c_j})\|^2}.
\end{equation}
Let $l_j$ be the length of $c_j$. The eigenvalues of $e^{i\theta_j}P_{c_j}$ are the $l_j$th roots of unity multiplied by $e^{i\theta_j}$. Consequently,
\begin{equation}
\begin{split}
	\min_{\theta_j} \|\Log(e^{i\theta_j} P_{c_j})\|^2
	&= \min_{\theta_j}\sum_{k=0}^{l_j-1} \Big(\theta_j + \frac{2\pi k}{l_j}\Big)^2\\
	&= \frac{\pi^2(l_j^2-1)}{3l_j}.
\end{split}
\end{equation}
It follows that 
\begin{equation}\label{eqb}
	\taumin 
	\leq \frac{\pi}{\sqrt{3}\omega}\sqrt{n-\sum_{k=1}^m \frac{1}{l_k}}.
\end{equation}

\begin{ex}\label{8dimexample}
Consider an $8$-dimensional system for which the spectra of $A$ and $\rho_i$ satisfy
\begin{alignat}{16}
	& a_1\ & <\ & a_2\ & <\ & a_3\ & <\ & a_4\ & <\ & a_5\ & <\ & a_6\ & =\ & a_7\ & <\ & a_8,\label{Aspec}\\
	& p_3\ & >\ & p_1\ & =\ & p_2\ & >\ & p_5\ & >\ & p_4\ & >\ & p_8\ & >\ & p_6\ & >\ & p_7.\label{Rspec}
\end{alignat}
In this case there are four passivizing permutations:
\begin{align}	
	&(1\;2\;3)(4\;5)(6\;7\;8),\label{perm1}\\
	&(1\;3)(2)(4\;5)(6\;7\;8),\label{perm2}\\
	&(1\;2\;3)(4\;5)(6)(7\;8),\label{perm3}\\
	&(1\;3)(2)(4\;5)(6)(7\;8).\label{perm4}
\end{align}
If we insert the lengths of the cycles of the first permutation into the right-hand side of \eqref{eqb}, we find that $\taumin\leq\pi\sqrt{41}/\omega\sqrt{18}$. And if we insert the lengths of the second or the third permutation's cycles, we find that $\taumin\leq\pi\sqrt{17}/3\omega$. Finally, if we insert the lengths of the fourth permutation's cycles, we find that $\taumin\leq\pi\sqrt{6}/2\omega$. 

The second and the third permutation can be obtained from the first by a division of the first and the last cycle, respectively. The division leads to a reduction of the upper bound in \eqref{eqb}. Similarly, the fourth permutation can be obtained by dividing both the first and the last cycle of the first permutation, which leads to an even greater improvement of the upper bound. (In fact, the upper bound $\pi\sqrt{6}/2\omega$ equals $\tauqsl$ and, hence, $\taumin$.) A division of the first and last cycle of the first permutation is possible because of the identities $p_1=p_2$ and $a_6=a_7$.
\end{ex}

Example \ref{8dimexample} shows that degeneracies in the spectrum of $A$ or $\rho_i$ sometimes make it possible to divide a cycle of a passivizing permutation into two cycles without changing the fact that the permutation is passivizing. Such a division always leads to a lowering of the upper bound in \eqref{eqb}.\footnote{The authors do not know whether from an arbitrary passivizing permutation one can always reach the passivization time by means of cycle division.} To see this suppose that $c=(k_1\,k_2 \dots k_{l})$ is a cycle of a passivizing permutation $\sigma$ and suppose that for some $i<j$ we have that $a_{k_i}=a_{k_j}$ or $p_{k_i}=p_{k_j}$. Then $c$ can be replaced by $(k_1 \dots k_i\, k_{j+1} \dots k_l)(k_{i+1} \dots k_j)$, see Figure \ref{cyclesplitting}, 
\begin{figure}[t]
	\centering
	\includegraphics[width=0.47\textwidth]{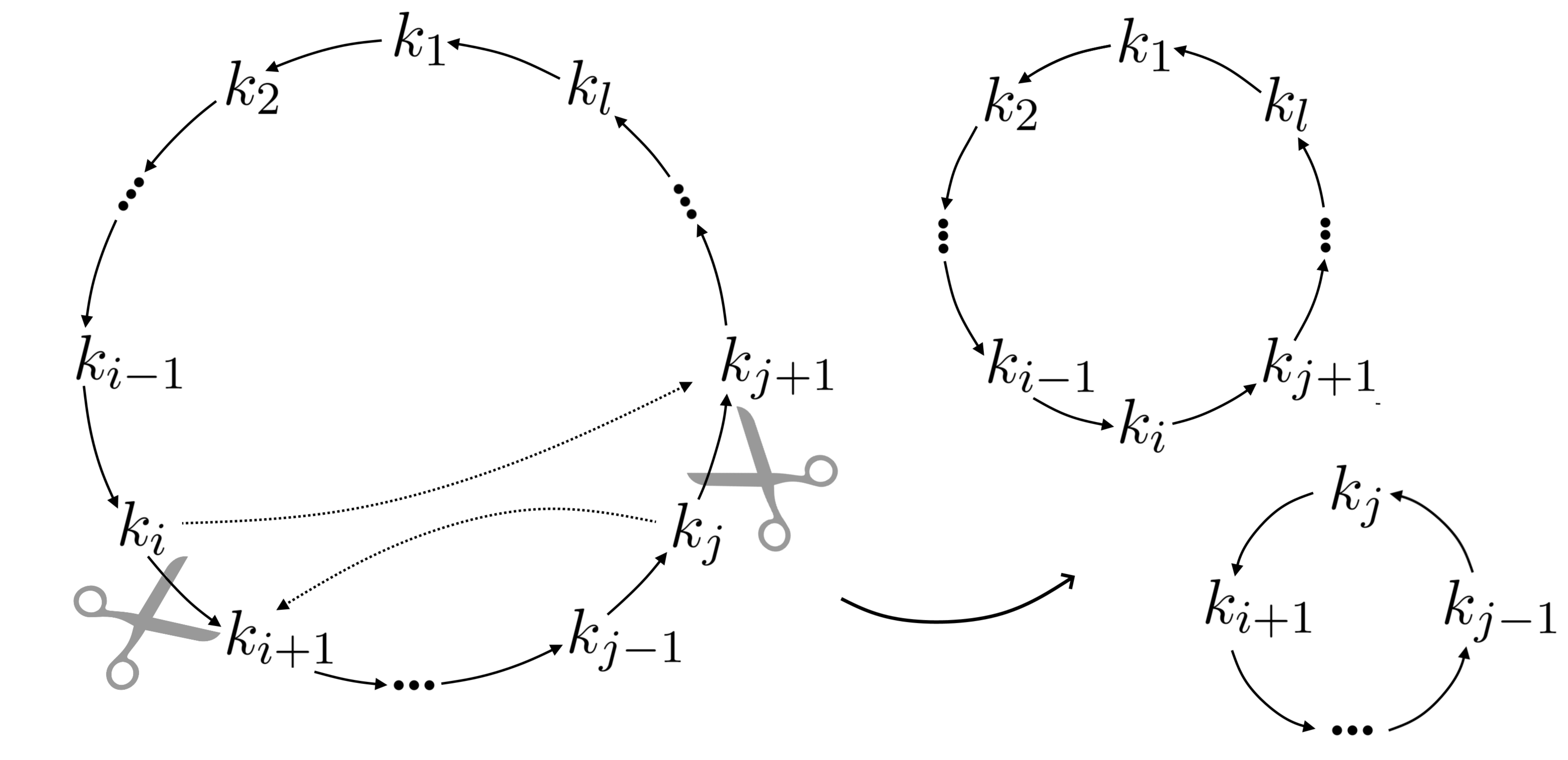}
	\caption{A cycle $(k_1\,k_2\dots k_l)$ of a passivizing permutation
	can be divided into two shorter cycles $(k_1
	\dots k_i\,k_{j+1}\dots k_l)$ and $(k_{i+1} \dots k_j)$ if $a_{k_i}=a_{k_j}$ or $p_{k_i}=p_{k_j}$.
	The resulting permutation is also passivizing.} 
	\label{cyclesplitting}
\end{figure}
and $\sigma$ be redefined accordingly. The resulting permutation is also passivizing. If the lengths of the two new cycles are $l_j'$ and $l_j''$, respectively, then 
\begin{equation}
	\frac{(l_j')^2-1}{l_j'}+\frac{(l_j'')^2-1}{l_j''}<\frac{l_j^2-1}{l_j}.
\end{equation}
The cycle division thus lowers the upper bound in \eqref{eqb}.

\subsection{Invariant subspaces of passivizing unitaries}\label{ispu}
If, as in Example \ref{8dimexample}, after repeated cycle division we end up with a passivizing involution, then we know from Proposition \ref{2cycles} that $\taumin=\tauqsl$. But if in Example \ref{8dimexample} we replace the identity $p_1=p_2$ in \eqref{Rspec} with $p_1>p_2$, then there are only two passivizing permutations, namely those in \eqref{perm1} and \eqref{perm3}, neither of which are involutions. The inequality in \eqref{eqb} guarantees that in this case, $\taumin$ is not greater than $\pi\sqrt{17}/3\omega$. However, none of the previous propositions certifies that the passivization time \emph{equals} $\pi\sqrt{17}/3\omega$. In this section and the next we will develop a method that can be used to prove that this actually is the case. The strategy is to break down the problem of determining $\taumin$ into a number of lower-dimensional problems which can be solved using results from the theory of generalized flag manifolds.

Let $\sigma$ be a passivizing permutation. Consider a decomposition of $\sigma$ into sub-permutations, $\sigma=\sigma_1\sigma_2\cdots\sigma_m$, where each sub-permutation $\sigma_j$ is a cycle or a product of cycles of $\sigma$. Define $\H_j$ as the linear span of the computational basis vectors whose labels are permuted by $\sigma_j$, and write $P_{\sigma_j}$ for the restriction of $P_\sigma$ to $\H_j$.

\begin{prop}\label{invariant}
If each eigenspace of $A$ and each eigenspace of $\rho_i$ is contained in an $\H_j$, then the $\H_j$s are invariant for the passivizing unitaries. Furthermore, the passivizing unitaries lying closest to $\1$ can be generated by time-optimal Hamiltonians that preserve the~$\H_j$s.
\end{prop}

\begin{proof}
The permutation operator $P_\sigma$ preserves the $\H_j$s. And so do the operators in the isotropy groups of $A$ and $\rho_i$. Proposition \ref{passivizing set} then tells us that all the passivizing unitaries preserve the $\H_j$s.

Let $U$ be any passivizing unitary operator at the minimum distance $d$ from $\1$. Since the $\H_j$s are mutually orthogonal and span $\H$, $U$ decomposes as a direct sum $U=\oplus_j U_j$ with $U_j$ acting on $\H_j$. Define $H$ as
\begin{equation}\label{defH}
	H=-\frac{i\omega}{d}\Log U=-\frac{i\omega}{d}\bigoplus_{j=1}^m \Log U_j.
\end{equation}
The Hamiltonian $H$ generates a shortest geodesic from $\1$ to $U$ and, hence, from $\1$ to $\P(\rho_i)$. Furthermore, $H$ satisfies the bounded bandwidth condition:
\begin{equation}
	\tr\big(H^2\big)=\frac{\omega^2}{d^2}\|\Log U\|^2=\omega^2.
\end{equation}
Thus, $H$ is a time-optimal Hamiltonian. By construction, $H$ preserves the spaces $\H_j$.
\end{proof}

\begin{center}
\emph{In the remainder of this section, we assume that the isotropy groups of $A$ and $\rho_i$ preserve the spaces $\H_j$.}
\end{center}

\noindent We define $\U(\H_j)_A$ as the group of unitary operators on $\H_j$ that preserve those eigenspaces of $A$ which are contained in $\H_j$. Similarly, we define $\U(\H_j)_{\rho_i}$ as the group of unitary operators on $\H_j$ that preserve those eigenspaces of $\rho_i$ which are contained in $\H_j$. Then 
\begin{equation}\label{fruit}
	\dist(\1,\P(\rho_i)) 
	= \sqrt{\sum_{j=1}^m \min_{U_j,V_j} \|\Log(U_jP_{\sigma_j} V_j)\|^2},
\end{equation}
where the minima are taken over all the $U_j$s and $V_j$s in $\U(\H_j)_A$ and $\U(\H_j)_{\rho_i}$, respectively. The summands on the right-hand side of \eqref{fruit} are generally difficult to calculate. But the problem simplifies somewhat if one of the isotropy groups contains the other. Because then we only need to minimize over the larger of the two: If $\U(\H_j)_A$ contains $\U(\H_j)_{\rho_i}$, then
\begin{equation}\label{RiA}
	\underset{U_j,V_j}{\min} \|\Log(U_jP_{\sigma_j} V_j)\| = \underset{U_j}{\min} \|\Log(P_{\sigma_j} U_j)\|,
\end{equation}
and if $\U(\H_j)_{\rho_i}$ contains $\U(\H_j)_A$, then
\begin{equation}\label{AiR}
	\underset{U_j,V_j}{\min} \|\Log(U_jP_{\sigma_j} V_j)\| = \underset{V_j}{\min} \|\Log(P_{\sigma_j} V_j)\|.
\end{equation}

The right-hand sides of \eqref{RiA} and \eqref{AiR} are geodesic distances in certain homogeneous spaces called generalized flag manifolds. Well-known examples of such are projective Hilbert spaces and Grassmann manifolds. In the next section, we describe how to calculate the distances on the right of \eqref{RiA} under certain circumstances. Notice that the isotropy group of $\rho_i$ is contained in that of $A$ if each eigenspace of $\rho_i$ is included in an eigenspace of $A$. This is the case, for example, if $\rho_i$ has a non-degenerate spectrum. The case when $\U(\H_j)_A$ is a subgroup of $\U(\H_j)_{\rho_i}$ can be treated similarly. Appendix \ref{flag} contains a brief review of generalized flag manifolds.

\subsection{Geodesic distance in generalized flag manifolds}
Let $\1_j$ be the identity operator on $\H_j$. Regard $\U(\H_j)$ as a manifold on which $\U(\H_j)_A$ acts from the right by operator pre-composition. If we equip $\U(\H_j)$ with the bi-invariant Riemannian metric that agrees with the Hilbert-Schmidt inner product on the Lie algebra of $\U(\H_j)$, then the right-hand side of \eqref{RiA} is the geodesic distance between the cosets of $\1_j$ and $P_{\sigma_j}$ in the geometry determined by the projected metric on the quotient manifold $\U(\H_j)/\U(\H_j)_A$:
\begin{equation}\label{avstand}
	\dist\big([\1_j],[P_{\sigma_j}]\big) 
	= \underset{U_j\in \U(\H_j)_A}{\min}\|\Log(P_{\sigma_j} U_j)\|.
\end{equation}
The quotient manifold $\U(\H_j)/\U(\H_j)_A$ is an example of a generalized flag manifold; see Appendix \ref{flag}. Next, we show how to calculate \eqref{avstand} under special circumstances.

\subsubsection{The Grassmann and Fubini-Study distances}
That the isotropy group of $A$ preserves $\H_j$ is equivalent to $\H_j$ being a direct sum of eigenspaces of $A$. If $\H_j$ is a sum of two eigenspaces of $A$, then $\U(\H_j)/\U(\H_j)_A$ is a Grassmann manifold, and \eqref{avstand} is the Grassmann distance between $[\1_j]$ and $[P_{\sigma_j}]$; see \cite{EdArSm1998} and Appendix~\ref{flag}. 

Suppose that $\H_j=\A_{k}\oplus\A_{k'}$. Let $n_{k}$ be the dimension of $\A_{k}$ and let $\Pi_{k}$ be the orthogonal projection of $\H_j$ onto $\A_{k}$. Furthermore, let $s_1,s_2,\dots,s_{n_{k}}$ be the singular values of $\Pi_{k}^\dagger P_{\sigma_j} \Pi_{k}$. Then, by \eqref{Gdistance} in Appendix \ref{flag}, 
\begin{equation}\label{avstandGgen}
	\dist\big([\1_j],[P_{\sigma_j}]\big) = \sqrt{ 2 \sum_{i=1}^{n_k} \big(\arccos\sqrt{s_i}\,\big)^2}.
\end{equation}
Since $P_{\sigma_j}$ is a permutation operator, each singular value is either $0$ or $1$. The number of $0$s equals the number of computational basis vectors in $\A_k$ that $P_{\sigma_j}$ maps into $\A_{k'}$. Let $\delta_j$ be twice this number. Formula \eqref{avstandGgen} yields
\begin{equation}\label{avstandG}
	\dist\big([\1_j],[P_{\sigma_j}]\big) = \frac{\pi\sqrt{\delta_j}}{2}.
\end{equation}

If one of $A$'s eigenspaces is $1$-dimensional, the Grassmann manifold $\U(\H_j)/\U(\H_j)_A$ is a projective Hilbert space. In this case, the Grassmann distance function in \eqref{avstandGgen} goes by the name ``the Fubini-Study distance.'' If $\ket{k}$ is the computational basis vector that spans the $1$-dimensional eigenspace of $A$, then
\begin{equation}\label{avstandFSgen}
	\dist\big([\1_j],[P_{\sigma_j}]\big)=\sqrt{2}\arccos{\braket{k}{\sigma(k)}}.
\end{equation}
Since $\braket{k}{\sigma(k)}=0$ or $\braket{k}{\sigma(k)}=1$, depending on whether $\sigma$ leaves $k$ invariant or not,
\begin{equation}\label{avstandFS}
	\dist\big([\1_j],[P_{\sigma_j}]\big) 
	= \begin{cases} 0 & \text{if $\sigma(k) = k$,} \\ \frac{\pi}{\sqrt{2}} & \text{if $\sigma(k)\ne k$.} \end{cases}
\end{equation}

\begin{rmk}
The reader might wonder if the assumption that the isotropy group of $A$ contains the isotropy group of $\rho_i$ is essential for the validity of \eqref{avstandG}. After all, in Example \ref{exempt}, we determined the passivization time for a system with a bivalent $A$ without making assumptions about the structure of $\U(\H)_{\rho_i}$, and the resemblance between \eqref{avstandG} and the quantum speed limit \eqref{theqsl} is striking. The problem is that if $\U(\H_j)_A$ does not contain $\U(\H_j)_{\rho_i}$, then \eqref{avstandG} does not produce a reliable contribution to the distance in \eqref{fruit}, in the sense that the right-hand side of \eqref{fruit} is independent of the passivizing permutation. Consider, for example, a qubit system prepared in the maximally mixed state and with a non-degenerated $A$. The permutations $\sigma=(1,2)$ and $\sigma'=(1)(2)$ are both passivizing. However, the distance from $[\1]$ to $[P_{\sigma}]$ in $\U(\H)/\U(\H)_A$ equals $\pi/\sqrt{2}$, while the distance from $[\1]$ to $[P_{\sigma'}]$ is $0$. Since $\rho_i$ is passive, it is the latter distance which equals $\dist(\1,\P(\rho_i))$.
\end{rmk}

\subsubsection{The flag distance}
If all the eigenspaces of $A$ in $\H_j$ are $1$-dimensional, then $\U(\H_j)/\U(\H_j)_A$ is a flag manifold \cite{Ar2003}. In this case
\begin{equation}\label{avstandF}
	\dist\big([\1_j],[P_{\sigma_j}]\big) = \sqrt{\frac{\pi^2}{3}\sum_{j=1}^m \frac{l_j^2-1}{l_j}},
\end{equation}
with $l_1,l_2,\dots,l_m$ being the lengths of the cycles of $\sigma_j$. The proof is the same as that of Proposition \ref{nondegenerate}.

\subsubsection{The generalized flag distance}
If $\H_j$ is a sum of more than two eigenspaces of $A$ and not all of these eigenspaces are $1$-dimensional, then $\U(\H_j)/\U(\H_j)_A$ is a generalized flag manifold \cite{Ar2003}. This manifold can be equipped with a metric such that for any pair of unitary operators $U$ and $V$ on $\H_j$, 
\begin{equation}\label{avstandGF}
	\hspace{-6pt}\dist\big([U],[V]\big)
	= \underset{U_j\in \U(\H_j)_A}{\min}\|\Log(U^\dagger V U_j)\|
\end{equation}
is the geodesic distance between the cosets of $U$ and $V$; see Appendix \ref{flag}. As far as the authors know, there is no closed formula for the geodesic distance \eqref{avstandGF}. However, an algorithm for numerically calculating \eqref{avstandGF} is proposed in the recent paper \cite{MaKiPe2020}.

\begin{ex}
Consider the $8$-dimensional system in Example \ref{8dimexample} but replace \eqref{Rspec} with the assumption
\begin{equation}
	p_3 > p_1 > p_2 > p_5 > p_4 > p_8 > p_6 > p_7.
\end{equation}
Then there are only two passivizing permutations:
\begin{align}	
	&(1\;2\;3)(4\;5)(6\;7\;8),\label{former}\\
	&(1\;2\;3)(4\;5)(6)(7\;8)\label{latter}.
\end{align}
Let $\sigma$ be the permutation in Equation \eqref{latter} and divide $\sigma$ into three sub-permutations: $\sigma_1=(1,2,3)$, $\sigma_2=(4,5)$, and $\sigma_3=(6)(7,8)$. The corresponding subspaces of $\H$ are the linear spans\footnote{The linear span, or ``sp'', of a set of vectors is the space consisting of all the linear combinations of the vectors in the set.}
\begin{alignat}{4}	
	&\H_1 &&=\ &&\linspan &&\{ \ket{1},\ket{2},\ket{3}\},\label{space1}\\
	&\H_2 &&=\ &&\linspan &&\{ \ket{4},\ket{5}\},\label{space2}\\
	&\H_3 &&=\ &&\linspan &&\{ \ket{6},\ket{7},\ket{8}\}.\label{space3}
\end{alignat}
The isotropy group of $A$ contains the isotropy group of $\rho_i$, and the the spaces in $\eqref{space1}-\eqref{space3}$ are direct sums of eigenspaces of $A$:
\begin{alignat}{3}	
	&\H_1 &&=\ &&\linspan\{\ket{1}\}\oplus\linspan\{\ket{2}\}\oplus\linspan\{\ket{3}\},\\
	&\H_2 &&=\ &&\linspan\{\ket{4}\}\oplus\linspan\{\ket{5}\},\\
	&\H_3 &&=\ &&\linspan\{\ket{6},\ket{7}\}\oplus\linspan\{\ket{8}\}.
\end{alignat}
Thus, we are in a situation to which Proposition \ref{invariant} applies. According to \eqref{avstandF} and \eqref{avstandFS}, $\dist([\1_1],[P_{\sigma_1}])=  \pi\sqrt{8}/3$ and $\dist([\1_2],[P_{\sigma_2}])=\dist([\1_3],[P_{\sigma_3}])=\pi/\sqrt{2}$. By Equations \eqref{fruit} and \eqref{RiA},
\begin{equation}
	\dist(\1,\P(\rho_i)) 
	= \sqrt{\frac{8\pi^2}{9} + \frac{\pi^2}{2} + \frac{\pi^2}{2}}=\frac{\pi\sqrt{17}}{3}.
\end{equation}
This confirms the claim regarding the passivization time in the first paragraph of Section \ref{ispu}.
\end{ex}

\begin{ex}
Consider a $14$-dimensional system for which the spectrum of $A$ is such that
\begin{equation}
	\hspace{-2pt}a_1\hspace{-1pt}<\hspace{-1pt}a_2\hspace{-1pt}=\hspace{-1pt}\cdots\hspace{-1pt}=\hspace{-1pt}a_5\hspace{-1pt}<\hspace{-1pt}a_6\hspace{-1pt}=\hspace{-1pt}a_7\hspace{-1pt}<\hspace{-1pt}\cdots\hspace{-1pt}<\hspace{-1pt}a_{12}\hspace{-1pt}=\hspace{-1pt}\cdots\hspace{-1pt}=\hspace{-1pt}a_{14}.
\end{equation}
Suppose that the isotropy group of $A$ contains the isotropy group of $\rho_i$, and suppose that $\sigma=\sigma_1\sigma_2\sigma_3$,~with
\begin{align}
	&\sigma_1=(6\;7\;10),\\
	&\sigma_2=(2\;12\;3\;14)(4\;13)(5),\\
	&\sigma_3=(1\;8\;9\;11),
\end{align}
is a passivizing permutation. The spaces
\begin{alignat}{4}	
	&\H_1 &&=\ &&\linspan &&\{ \ket{6},\ket{7},\ket{10}\},\\
	&\H_2 &&=\ &&\linspan &&\{ \ket{2},\ket{3},\ket{4},\ket{5},\ket{12},\ket{13},\ket{14}\},\\
	&\H_3 &&=\ &&\linspan &&\{ \ket{1},\ket{8},\ket{9},\ket{11}\},
\end{alignat}
are invariant for all the passivizing unitaries, and they can be decomposed into eigenspaces of $A$ as follows
\begin{alignat}{3}	
	&\H_1 &&=\ &&\linspan\{\ket{6},\ket{7}\}\oplus\linspan\{\ket{10}\},\label{dec1}\\
	&\H_2 &&=\ &&\linspan\{\ket{2},\ket{3},\ket{4},\ket{5}\}\oplus\linspan\{\ket{12},\ket{13},\ket{14}\},\label{dec2}\\
	&\H_2 &&=\ &&\linspan\{\ket{1}\}\oplus\linspan\{\ket{8}\}\oplus\linspan\{\ket{9}\}\oplus\linspan\{\ket{11}\}.\label{dec3}
\end{alignat}
The decompositions tell us that $\U(\H_1)/\U(\H_1)_A$ is a projective Hilbert space, that $\U(\H_2)/\U(\H_2)_A$ is a Grassmann manifold, and that $\U(\H_3)/\U(\H_3)_A$ is a flag manifold. By \eqref{avstandFS}, \eqref{avstandG}, and \eqref{avstandF},
\begin{align}
	&\dist\big([\1_1],[P_{\sigma_1}]\big) = \frac{\pi}{\sqrt{2}},\\
	&\dist\big([\1_2],[P_{\sigma_2}]\big) = \frac{\pi\sqrt{3}}{\sqrt{2}},\\
	&\dist\big([\1_3],[P_{\sigma_3}]\big) = \frac{\pi\sqrt{5}}{2}.
\end{align}
Hence, by \eqref{fruit}, 
\begin{equation}\label{sss}
	\dist\big(\1,\P(\rho_i)\big) 
	= \sqrt{ \frac{\pi^2}{2} + \frac{3\pi^2}{2} + \frac{5\pi^2}{4}}
	= \frac{\pi\sqrt{13}}{2}.
\end{equation}
\end{ex} 

Above, we have treated the case when each eigenspace of $\rho_i$ in $\H_j$ is a subspace of an eigenspace of $A$. As mentioned at the end of Section \ref{ispu}, the case when each eigenspace of $A$ in $\H_j$ is contained in an eigenspace of $\rho_i$ can be treated in the same way. Also, the hybrid case when each eigenspace of $\rho_i$ is a subspace of an eigenspace of $A$ or is a sum of eigenspaces of $A$ can be handled similarly. Then the generalized flag manifold is given by the quotient of $\U(\H_j)$ and the group generated by the union of $\U(\H_j)_A$ and $\U(\H_j)_{\rho_i}$.

\section{On the power of energy extraction from quantum batteries}\label{charging}
In this final section, we use results from previous sections to derive bounds on the power with which energy can be reversibly extracted from a quantum battery. We follow \cite{AlBaNi2004,AlFa2013,BiViMoGo2015,CaPoBiCeGoViMo2017} and define a quantum battery as a closed $n$-dimensional quantum system that can release energy through a controllable process, causing the battery state to change according to a von Neumann equation of the form $\dot\rho(t)=-i[H + V(t),\rho(t)]$. Here, $H$ is the battery's internal Hamiltonian, and $V(t)$ is a time-dependent potential. We limit our considerations to cyclic processes and, thus, assume that the potential vanishes outside a finite time interval $[0,\tau]$, the final time $\tau$ being the \emph{duration} of the process. Also, to connect with previous sections, we assume that the initial state of the battery is incoherent relative to $H$ and that the available resources are limited in such a way that the bandwidth of the potential cannot exceed a given value:
\begin{equation}\label{potentialbbc}
	\tr\big(V(t)^2\big)\leq \omega^2.
\end{equation}
The internal Hamiltonian $H$ here plays the role of the observable $A$, and minimality of $\EE_H$ characterizes the passivity of states. We prefer to denote the eigenvalues of $H$ by $\epsilon_k$ rather than $a_k$, but apart from that we follow the standard set in Section \ref{puig}.

\subsection{Ergotropy and the power of complete discharge processes}
The maximal amount of (average) energy that can be cyclically extracted from a battery state $\rho_i$ is called the ergotropy of the battery \cite{AlBaNi2004}. The ergotropy equals the difference in energy of $\rho_i$ and that of a passive state:
\begin{equation}
	W(\rho_i)=\EE_H(\rho_i)-\EE_H(\rho_p).
\end{equation}
Since all the passive states have the same energy content, the ergotropy only depends on the initial state and the internal Hamiltonian. In terms of the internal energies and the eigenvalues of $\rho_i$, the ergotropy reads 
\begin{equation}
	W(\rho_i)=\sum_{k=1}^n \epsilon_k\big(p_k-p_{\sigma(k)}\big).
\end{equation}
Here $\sigma$ is any passivizing permutation.

No energy can be extracted through a cyclic unitary process from a battery in a passive state \cite{AlBaNi2004}. Therefore we call an energy extraction process that leaves the battery in a passive state \emph{a complete discharging} of the battery. The (average) power of a complete discharging of duration $\tau$ is $W(\rho_i)/\tau$. We define $\taumin$ as the passivization time of the battery determined by the bandwidth condition \eqref{bbcondition} with the same right-hand side as in \eqref{potentialbbc}. Note that this definition of the passivization time does not take any characteristics of the internal Hamiltonian $H$ into account, even though $H$ affects the battery's dynamics. If the bandwidth of the internal Hamiltonian greatly exceeds $\omega^2$, then no Hamiltonian of the form $H + V(t)$ is even close to being time-optimal. Nevertheless, according to the next proposition, the duration of a complete discharging is at least $\taumin$. Thus, the power of such a process is bounded from above by $W(\rho_i)/\taumin$.

\begin{prop}\label{upperpower}
The duration of a complete discharge process is greater than the passivization time.
\end{prop}

\begin{proof}
Let $V(t)$ be a potential that satisfies \eqref{potentialbbc} and completely discharges the battery in the time $\tau$. We regard the potential as a perturbation and go over to the interaction picture. In the interaction picture, the state of the battery evolves according to $\dot\rho_I(t)=-i[V_I(t),\rho_I(t)]$ with $V_I(t)=e^{itH}V(t) e^{-itH}$. The bandwidth of $V_I(t)$ is upper bounded by $\omega^2$ since $\tr(V_I(t)^2)=\tr(V(t)^2)$, and the final state is passive since the passive states commute with $H$. From this follows that $\tau\geq\taumin$.
\end{proof}

The next two examples are direct consequences of Proposition \ref{upperpower}, Example \ref{active}, and Proposition \ref{nondegenerate}.

\begin{ex}
Consider a fully charged battery. That is, consider a battery in a maximally active state $\rho_i$. Let $P$ be the power of a complete discharge process with a potential whose bandwidth is bounded by $\omega^2$. Then
\begin{equation}
	P\leq \frac{\omega}{\pi} \sqrt{\frac{2}{m}} \sum_{k=1}^n \epsilon_k \big(p_k-p_{n-k+1}\big),
\end{equation}
where $m$ is the greatest integer for which $p_m<p_{n-m+1}$ and $\epsilon_m<\epsilon_{n-m+1}$ hold.
\end{ex}

\begin{ex}\label{tretton}
Suppose that the internal Hamiltonian and the initial state of a battery are non-degenerate. Let $P$ be the power of a complete discharging of the battery with a potential whose bandwidth is bounded by $\omega^2$. Then
\begin{equation}
	P\leq  \frac{\omega}{\pi} \sqrt{\frac{3}{n-\sum_{k=1}^m\frac{1}{l_k}}}  W(\rho_i).
\end{equation}
Here $l_1,l_2,\dots,l_m$ are the lengths of the cycles of the unique permutation $\sigma$ that passivizes the initial state.
\end{ex}

Since we require that the discharge processes are cyclic, their duration can never be as small as the passivization time. However, there are discharge processes whose duration comes arbitrarily close to the passivization time.

\begin{prop}\label{tight_bound}
The passivization time is a tight bound on the duration of complete discharge processes. 
\end{prop}

\begin{proof}
Let $\epsilon>0$ be arbitrary and let $u(t)$ be a smooth function which equals $0$ for $t\leq 0$ and $\taumin$ for $t\geq \taumin+\epsilon$ and whose derivative $u'(t)$ takes values between $0$ and $1$ for all $t$. Let $V_I$ be a time-optimal Hamiltonian and define a potential as $V(t)=u'(t) e^{-itH}V_I e^{itH}$. The potential vanishes outside $[0,\taumin+\epsilon]$ and satisfies the bounded bandwidth condition \eqref{potentialbbc}: $\tr(V(t)^2)=u'(t)^2\omega^2\leq \omega^2$. Furthermore, the solution $\rho(t)$ to the von Neumann equation with Hamiltonian $H+V(t)$ that extends from $\rho_i$ is passive at $t=\taumin+\epsilon$. Explicitly, $\rho(\taumin+\epsilon)=e^{-i\taumin V_I}\rho_i e^{i\taumin V_I}$, which is passive by assumption.
\end{proof}

\subsection{Assisted discharging}
In an assisted discharging, a catalyst is used in the discharge process. The catalyst is ultimately unchanged but may interact and become correlated with the battery during the discharge process. According to Proposition \ref{assistedlowerbound}, the power $P$ with which a fully charged battery can be discharged with the help of an $n_c$-dimensional catalyst is bounded from above as
\begin{equation}
	P \leq \frac{2\omega}{\pi}\sqrt{\frac{n_c}{\delta}}
	W(\rho_i).
\end{equation} 
Here $\delta$ is the discrepancy of the fully charged battery.

\subsection{Collective discharging.}\label{multidischarging}
We define a \emph{battery ensemble} to be a system built up of several identical batteries. In a collective complete discharge process, the states of all the batteries are collectively transformed into uncorrelated passive states.  

Let $P_\sharp$ be the supremum of all powers of collective complete discharge processes directed by global potentials whose bandwidth is bounded from above by $\omega^2Nn^{N-1}$, and let $P_{\|}$ be the supremum of all powers of parallel complete discharge processes directed by local potentials fulfilling \eqref{potentialbbc}. The fraction $P_\sharp/P_{\|}$ is a measure of the advantage of using a collective complete discharge process over a parallel complete discharge process; see \cite{CaPoBiCeGoViMo2017}. According to Proposition \ref{upperpower} and \eqref{Nqsl}, $P_\sharp\leq W(\rho_i^{\otimes N})/\tauNqsl$, and by Proposition \ref{tight_bound}, $P_{\|}=W(\rho_i^{\otimes N})/\taumin$. The advantage is thus upper bounded by the fraction of the single battery passivization time and the collective quantum speed limit:
\begin{equation}
    \frac{P_\sharp}{P_{\|}}
    \leq \frac{\taumin}{\tauNqsl}
    =\frac{2\omega\taumin}{\pi}\sqrt{\frac{Nn^{N-1}}{\delta_N}}.
\end{equation}

\begin{ex}
If the batteries of the ensemble are prepared in mixed and maximally active qubit states, then, by Example \ref{collectivequbit}, the power $P$ with which energy can be extracted in a complete discharge process governed by a potential satisfying \eqref{collectivebbcondition} is bounded as
\begin{equation}
    P\leq \frac{\sqrt{2}\omega(\epsilon_2-\epsilon_1)(p_2-p_1)}{\pi}
    \sqrt{\frac{2^N - \wp(N)\binom{N}{N/2}}{N2^N}}.
\end{equation}
Furthermore, the advantage of a collective discharging is
\begin{equation}
    \frac{P_\sharp}{P_{\|}}=\sqrt{\frac{N2^N}{2^N-\wp(N)\binom{N}{N/2}}}.
\end{equation}
If the batteries are identically prepared maximally active full rank qutrits, then, by Example \ref{collectivequtrit}, the power with which energy can be extracted in a complete discharge process is upper bounded according to
\begin{equation}
    P\leq \frac{2\omega(\epsilon_3-\epsilon_1)(p_3-p_1)}{\pi\sqrt{3}}
    \sqrt{\frac{N3^N}{3^N-\sum_{k=0}^{\lfloor N/2\rfloor}\binom{N}{k,k}}}.
\end{equation}
In this case, the advantage of a collective discharging is
\begin{equation}
    \frac{P_\sharp}{P_{\|}}= \sqrt{\frac{2N3^{N-1}}{3^N-\sum_{k=0}^{\lfloor N/2\rfloor} \binom{N}{k,k}}}.
\end{equation}
\end{ex}

We saw in Example \ref{neithernor} that the product of two passive qutrit states need not be globally passive. In the next example, we compare the power of an optimal collective passivization process with the power of an optimal global passivization process for such a qutrit battery.

\begin{ex}
Suppose that $H$ is non-degenerate with energies satisfying
\begin{equation}\label{spec1}
	2\epsilon_1<\epsilon_1+\epsilon_2<2\epsilon_2<\epsilon_1+\epsilon_3<\epsilon_2+\epsilon_3<2\epsilon_3,
\end{equation}
and suppose that $\rho_i$ is a maximally active qutrit state with a spectrum that satisfies 
\begin{equation}\label{spec2}
	p_1^2<p_1p_2<p_2^2<p_1p_3<p_2p_3<p_3^2.
\end{equation}
Each battery has a unique passive state $\rho_p$. The product $\rho_p^{\otimes 2}$, however, is not a passive state; see Example \ref{neithernor}. The energy one extracts from $\rho_i^{\otimes 2}$ in a collective passivization process is less than the energy one extracts in a global passivization process. To be precise, the difference in energy content of $\rho_p^{\otimes 2}$ and that of a globally passive state is $(\epsilon_1+\epsilon_3-2\epsilon_2)(p_1p_3-p_2^2)$. There are collective passivization processes with a duration arbitrarily close to $\pi/2\omega$ and global passivization processes with a duration arbitrarily close to (but not less than) $\pi/\omega\sqrt{3}$. Thus, the difference in `optimal power' of a collective passivization process and a global passivization process is
\begin{equation}\label{diff}
\begin{split}
	\frac{2\omega}{\pi}(2-\sqrt{3}&)(\epsilon_3-\epsilon_1)(p_3-p_1)\\
		&-\frac{\omega\sqrt{3}}{\pi}(\epsilon_1+\epsilon_3-2\epsilon_2)(p_1p_3-p_2^2).
\end{split}
\end{equation}
Due to \eqref{spec1} and \eqref{spec2}, the expression in \eqref{diff} is positive. Hence, an optimal collective passivization process is more `powerful' than a global passivization process.
\end{ex}

\subsection{Energy and power fluctuations of discharge processes}
We finish this paper with some observations concerning energy fluctuations of discharge processes. We consider a battery prepared in a state of definite energy through a stochastic preparation procedure. Such a statistical state can be modeled by a density operator that is incoherent relative to the internal Hamiltonian. The average energy extracted from such a battery in a complete discharge process, which leaves the battery in a passive statistical state, is equal to the ergotropy. 

To derive an expression for the variation in transferred energy in a complete discharge process, let $\rho_i$ be the prepared statistical state, let $H$ be the internal Hamiltonian, and let $E_1,E_2,\dots, E_r$ be the different, possibly degenerate, eigenvalues of $H$. Also, let $\Pi_k$ be the orthogonal projection onto the eigenspace corresponding to $E_k$. The variation in transferred energy is
\begin{equation}\label{variation}
	\Delta^2 W(\rho_i) = \sum_{k,l=1}^r (E_k-E_l)^2 p(l,k)- W(\rho_i)^2,
\end{equation}
where $p(l,k)$ is the probability that a battery starts in a state with energy $E_k$ and ends up in a state with energy $E_l$ in a complete discharge process. If $U$ is the unitary implemented by the discharge process, then 
\begin{equation}
	p(l,k) = \tr(\Pi_l U \Pi_k \rho_i \Pi_k U^\dagger).
\end{equation} 
If we insert this expression into \eqref{variation} and simplify, we obtain the following expression for the variation
\begin{equation}\label{variation2}
\begin{split}
 \Delta^2W&(\rho_i) 
  = \Delta^2H(\rho_i) + \Delta^2H(\rho_p)\\ 
    &  + 2\,\EE_H(\rho_i)\EE_H(\rho_p) 
    - 2\tr\big(U^\dagger HUH\rho_i\big).
\end{split}
\end{equation}
Here $\Delta^2H(\rho)=\EE_{H^2}(\rho)-\EE_H^2(\rho)$, and $\rho_p=U\rho_i U^\dagger$ is the final statistical state of the battery. Since $U$ does not commute with $H$, unless $\rho_i$ is passive, the last term is process-dependent (while the other terms are not). Thus, the variation in the amount of energy extracted from a battery may differ for different complete discharge processes. We nuance this statement with a proposition and an example:

\begin{prop}\label{fluctuation}
If the unitaries in the isotropy group of the initial state commute with $H$, then all complete discharge processes have the same variation in the transferred energy.  
\end{prop}

\begin{proof}
Let $P$ be any passivizing unitary and let $\sigma$ be any passivizing permutation. By Proposition \eqref{passivizing set}, $P$ can be decomposed as $P=UP_\sigma V$, where $U$ commutes with $H$ and $V$ commutes with $\rho_i$. By assumption, $V$ also commutes with $H$. A straightforward calculation yields
\begin{equation}
	\tr\left(P^\dagger HPH\rho_i\right)=\tr\left(P_\sigma^\dagger HP_\sigma H\rho_i\right).
\end{equation}
This shows that the value of the final term in \eqref{variation2} is independent of $P$ and, hence, that the variation in transferred energy is process-independent.
\end{proof}

According to Proposition \ref{fluctuation}, all complete discharge processes of a battery in a non-degenerate statistical state have the same variation in the energy transfer. The following example shows that this need not be the case if the state has a degenerate spectrum.

\begin{ex}\label{qutritvar}
Consider a qutrit battery whose internal Hamiltonian is non-degenerate with spectrum $\epsilon_1=0<\epsilon_2<\epsilon_3$, and whose initial statistical state is degenerate with spectrum $p_1=p_2<p_3$. Relative to the computational basis, a general passivizing unitary has the matrix 
\begin{equation}
    U=\begin{bmatrix} 
        0                   & 0                                 & e^{i\alpha} \\ 
        \sqrt{a}e^{i\beta}  & \sqrt{1-a}e^{i(\beta+\theta)}     & 0 \\ 
        \sqrt{1-a}e^{i\gamma}  & -\sqrt{a}e^{i(\gamma+\theta)}  & 0 
        \end{bmatrix},
\end{equation}
where $\alpha,\beta,\gamma$, and $\theta$ are arbitrary and $0\leq a \leq 1$. The variation in transferred energy is independent of the~phases:
\begin{equation}
    \Delta^2W(\rho_i) = p_1\epsilon_3^2(9p_1+5)-2p_1a(\epsilon_3-\epsilon_2).
\end{equation}
Also, $\Delta^2W(\rho_i)$ is maximal for $a=0$ and minimal for $a=1$. We find time-optimal processes among those that implement a $U$ for which $a=0$. But none of the processes that implement a $U$ for which $a=1$ is time-optimal.
\end{ex}

Example \ref{qutritvar} shows that there may be a trade-off between being time-optimal and having small fluctuations in transferred energy for complete discharge processes. Further investigation is required. 

\section{Summary and outlook}
In this paper, we have considered, and to some extent answered, the following question: Suppose that a finite-dimensional quantum system is prepared in a state that is incoherent relative to an observable. In how short a time can the state be transformed into a passive state for the observable provided that the Hamiltonian responsible for the transformation has bounded bandwidth?

We began by deriving a general QSL for the transformation time, which we also expanded to a lower time-bound for collective passivization processes. We then showed that for some systems, such as systems prepared in maximally active states, the QSL is equal to the passivization time, that is, the least possible time in which the system can be transformed into a passive state under the specified conditions. But we also showed that for some systems, the QSL is not tight. We calculated the passivization time explicitly for systems such that the observable and the initial state are non-degenerate. Then we developed a method to determine the passivization time for degenerate systems. The method presupposes that the eigenspaces of the observable and the state are in a particular relative constellation, which means that the method does not apply to all conceivable systems that match the description above - other approaches are required.

The problem discussed in this paper is an example of a brachistochrone problem \cite{CaHoKoOk2006,CaHoKoOk2007,CaHoKoOk2008,RuSt2015,WaAlJaLlLuMo2015,WaKo2020}. For constrained closed quantum systems, such problems are typically reformulated as one or more time-local relations for the  Hamiltonian using the calculus of variations \cite{CaHoKoOk2006,CaHoKoOk2007,WaAlJaLlLuMo2015,WaKo2020}. However, such relations are generally of little help in determining the shortest possible transformation time. For example, in the case dealt with here, calculus of variations would generate the results in Propositions \ref{time-independent} and \ref{parallel}. From there, there is still a long way to go. 

In the last section, we applied the results from previous sections to quantum batteries. Specifically, we derived upper bounds on the power with which energy can be extracted from a quantum battery through a cyclic unitary process. Here we only considered complete discharge processes, that is, processes that leave the battery in a passive state relative to the battery's internal Hamiltonian. From such a state, no more energy can be extracted through a cyclic unitary process.

Recently, the interest in quantum batteries has grown considerably, not least because of their predicted practical significance \cite{BiCoGoAnAd2019}. Here we have assumed that the battery is initially in an incoherent state relative to the internal Hamiltonian. The next step is to extend the results to quantum batteries in coherent states. The recent paper \cite{FrBiGuMiGoPl2020} is possibly the first step in such a direction.

\section{Acknowledgments}
The authors thank Supriya Krishnamurthy for the many inspiring discussions and Ingemar Bengtsson for suggesting improvements to the text.

\onecolumn
\newpage

\appendix

\section{Proof of Proposition \ref{incoherent}}\label{E}
Let $\rho$ be a passive or a maximally active state. Since the passive and the maximally active states are extremal for $\EE_A$, the differential of $\EE_A$ vanishes at $\rho$. Any tangent vector of $\S(\rho_i)$ at $\rho$ can be represented as $-i[B,\rho]$, where $B$ is some Hermitian operator. Conversely, any such $B$ represents a tangent vector of $\S(\rho_i)$ at $\rho$. We have that 
\begin{equation}
	\tr\big(B(-i[\rho,A])\big)
	= \tr\big(A(-i[B,\rho])\big)\\
	= d\EE_A\big(-i[B,\rho]\big)\\	
	= 0.
\end{equation}
Since this holds for every Hermitian $B$, $[\rho,A]=0$.

\section{Proof of Proposition \ref{saturates}}\label{S}

Suppose that $H(t)$ is a time-optimal Hamiltonian. Let $f(t)$ be a continuous and everywhere positive function such that $\tr(H(t)^2)\leq f(t)\leq \omega^2$. For $0\leq t\leq \taumin$, define
\begin{equation}
	\tau(t)=\frac{1}{\omega}\int_0^t \ds \sqrt{f(s)}.
\end{equation}
Since $f$ is continuous and positive, $\tau(t)$ is differentiable and monotonically increasing. Thus $\tau(t)$ is invertible and has a differentiable inverse $t(\tau)$. The domain of $t(\tau)$ is the interval $[0,\tau(\taumin)]$. At any $\tau$ in this interval,
\begin{equation}\label{uppskattning}
	\frac{\d t}{\d \tau}(\tau) = \frac{\omega}{\sqrt{f(t(\tau))}}.
\end{equation}

Define $\varrho(\tau)=\rho(t(\tau))$ and set $K(\tau)=\frac{\d t}{\d \tau}(\tau)H(t(\tau))$. The curve of states $\varrho(\tau)$ emanates from $\rho_i$, arrives at $\rho(\taumin)$ at the time $\tau(\taumin)$, and is generated by $K(\tau)$:
\begin{equation}
		\frac{\d \varrho}{\d \tau}(\tau)
		= -i[H(t(\tau)),\rho(t(\tau))]\frac{\d t}{\d \tau}(\tau)\\
		= -i[K(\tau),\varrho(\tau)].
\end{equation}
$K(\tau)$ satisfies the bounded bandwidth condition: 
\begin{equation}
	\tr\big(K(\tau)^2\big) = \frac{\omega^2}{f(t(\tau))}\tr\big(H(t(\tau))^2\big) \leq \omega^2.
\end{equation}1
If $\tr(H(t)^2)<\omega^2$ for some $t$, we can choose $f$ to be strictly less than $\omega^2$ in a neighborhood of that $t$. But then 
\begin{equation}
	\tau(\taumin)=\frac{1}{\omega}\int_0^{\taumin} \ds \sqrt{f(s)}<\taumin.
\end{equation}
This contradicts the assumption that $\taumin$ lower bounds the transformation time. Hence $\tr(H(t)^2)=\omega^2$ for all $t$.

\section{Proof of Proposition \ref{shortest}}\label{U}
According to Proposition \ref{saturates}, we need to show that no smooth curve $U(t)$ that extends from $\1$ to $\P(\rho_i)$ has a length shorter than $\omega\taumin$. Reparameterize $U(t)$ to have constant speed $\omega$ and let $\tau$ be the time when $U(t)$ arrives at $\P(\rho_i)$. The Hamiltonian $H(t)=-i\dot U(t)U(t)^\dagger$ generates $U(t)$ and satisfies the condition \eqref{bbcondition}. But then, by the definition of the passivization time, $\tau\geq\taumin$. The length of $U(t)$, which is $\omega\tau$, is therefore not less than $\omega\taumin$.

\section{The principal logarithm}\label{L}
The principal logarithm assigns a skew-Hermitian operator to each unitary operator on $\H$. It is defined as follows. Let $U$ be a unitary operator on $\H$ and let $e^{i\theta_1}, e^{i\theta_2},\dots,e^{i\theta_n}$ be the eigenvalues of $U$ with corresponding eigenvectors $\ket{\psi_1},\ket{\psi_2},\dots,\ket{\psi_n}$. Choose the phases $\theta_j$ in \emph{the principal branch} $(-\pi,\pi]$. Then
\begin{equation}
	\Log U=\sum_{k=1}^n i\theta_k\ketbra{\psi_k}{\psi_k}.
\end{equation}
The principal logarithm satisfies $e^{\Log U}=U$. But for a skew-Hermitian operator $S$, $\Log(e^{S})=S$ if, and only if, the imaginary parts of the eigenvalues of $S$ belong to the principal branch.

\section{Derivation of Equation \ref{logdistance}}\label{G}
Since the metric is bi-invariant, so is the geodesic distance. Therefore, it suffices to show that $\dist(\1,U)=\|\Log U\|$ or, which is equivalent, that $e^{-it\Log U}$ is a shortest geodesic from $\1$ to $U$. Let $e^{i\theta_1},e^{i\theta_2},\dots,e^{i\theta_n}$ be the eigenvalues of $U$ and select the phases of these in the principal branch $(-\pi,\pi]$. Also, let $e^{-itH}$ be any geodesic from $\1$ to $U$ such that $U=e^{-iH}$. Write $\epsilon_1,\epsilon_2,\dots,\epsilon_n$ for the eigenvalues of $H$. Possibly after a re-indexing, $e^{-i\epsilon_j}=e^{i\theta_j}$ and, hence,  $\epsilon_j=-\theta_j\bmod{2\pi}$. This implies that $\theta_j^2\leq \epsilon_j^2$. The length of $e^{-itH}$ equals $\|H\|$, and the length of $e^{-it\Log U}$ equals $\|\Log U\|$. The inequality 
\begin{equation}
	\|\Log U\|^2=\sum_{j=1}^n \theta_j^2\leq \sum_{j=1}^n \epsilon_j^2 = \|H\|^2
\end{equation}
shows that $e^{-it\Log U}$ is not longer than $e^{-itH}$.

\section{Proof of Proposition \ref{passivizing set}}\label{P}
That $UPV$ belongs to $\P(\rho_i)$ for each $U$ in $\U(\H)_A$ and each $V$ in $\U(\H)_{\rho_i}$ follows from $\U(\H)_A$ leaving the set of passive states invariant. Conversely, that every passivizing unitary has the form  $UPV$ for some $U$ in $\U(\H)_A$ and $V$ in $\U(\H)_{\rho_i}$ follows from $\U(\H)_A$ acting transitively on the set of passive states. Indeed, if $W$ is any passivizing unitary, then $W\rho_i W^\dagger=UP\rho_i P^\dagger U^\dagger$ for some $U$ in $\U(\H)_A$ due to the transitivity of the action of $\U(\H)_A$. But then $P^\dagger U^\dagger W\rho_i=\rho_i P^\dagger U^\dagger W$ which shows that $W=UPV$ for some $V$ in $\U(\H)_{\rho_i}$. 

That $\P(\rho_i)$ is a submanifold of $\U(\H)$ follows from $\P(\rho_i)$ being the orbit of $P$ under the action $(U,V)\cdot W=UWV^\dagger$ of $\U(\H)_A \times \U(\H)_{\rho_i}$ on $\U(\H)$: Since the product of the isotropy groups is compact, the action is proper and, hence, the orbits are submanifolds of $\U(\H)$.

\section{Flag manifolds}\label{flag}
Let $\U(n)$ be the group of unitary $n\times n$ matrices and let $n_1,n_2,\dots, n_l$ be any sequence of positive integers that add up to $n$. Write $\U(n_1,n_2,\dots, n_l)$ for the subgroup of $\U(n)$ consisting of the block-diagonal matrices whose blocks have dimensions $n_1\times n_1, n_2\times n_2,\dots,n_l\times n_l$. Let $\U(n_1,n_2,\dots, n_l)$ act from the right on $\U(n)$ by matrix multiplication. The quotient manifold $\U(n)/\U(n_1,n_2,\dots, n_l)$ is called a generalized flag manifold \cite{Ar2003}. Grassmann manifolds $\U(n)/\U(n_1,n_2)$ and the projective Hilbert space $\U(n)/\U(1,n-1)$ are special cases. A flag manifold is a generalized flag manifold for which all the $n_j$s equal $1$.

We write $[U]$ for the coset of a unitary matrix $U$. The coset is the element of the generalized flag manifold which is represented by $U$. The quotient projection $\p(U)=[U]$ is a principal fiber bundle with fiber $\U(n_1,n_2,\dots, n_l)$; see \cite{KoNo1996}. We equip $\U(n)$ with the bi-invariant metric $g$ that at the Lie algebra is given by $g(X,Y)=-\tr(XY)$, and we equip $\U(n)$ with the connection form whose kernel at a $U$ is the orthogonal complement of the tangent space of the fiber $\p^{-1}([U])$ at $U$. Then there is a unique metric on  $\U(n)/\U(n_1,n_2,\dots,n_l)$ which turns $\p$ into a Riemannian submersion \cite{Mo2002}. 

The geodesic distance between $[U]$ and $[V]$ equals the length of the shortest curve in $\U(n)$ connecting the fibers $\p^{-1}([U])$ and $\p^{-1}([V])$. For Grassmann manifolds, there is a formula for the geodesic distance: Let $s_1,s_2,\dots,s_n$ be the singular values of $U^\dagger V$. Then
\begin{equation}\label{Gdistance}
	\dist\big([U],[V]\big) = \sqrt{ 2 \sum_{i=1}^{n} \big(\arccos\sqrt{s_i}\,\big)^2},
\end{equation}
see \cite{EdArSm1998}. For generalized flag manifolds beyond Grassmann manifolds, no formula for the geodesic distance is known. However, an algorithm for numerically calculating the geodesic distance has been proposed recently~\cite{MaKiPe2020}.

\end{document}